\documentclass[10pt,journal,compsoc]{IEEEtran}
\ifCLASSOPTIONcompsoc
  \usepackage[nocompress]{cite}
\else
  \usepackage{cite}
\fi
\ifCLASSINFOpdf
   \usepackage[pdftex]{graphicx}
   \DeclareGraphicsExtensions{.pdf,.jpeg,.png}
\else
\fi
\usepackage{amsmath}
\usepackage{amsthm}
\usepackage{amsbsy}
\usepackage{algorithm}
\usepackage{algpseudocode}
\algnewcommand\algorithmicin{\textbf{Input}}

\usepackage{array}

\newcommand{\PreserveBackslash}[1]{\let\temp=\\#1\let\\=\temp}
\newcolumntype{C}[1]{>{\PreserveBackslash\centering}p{#1}}
\newcolumntype{R}[1]{>{\PreserveBackslash\raggedleft}p{#1}}
\newcolumntype{L}[1]{>{\PreserveBackslash\raggedright}p{#1}}

\ifCLASSOPTIONcompsoc
  \usepackage[caption=false,font=footnotesize,labelfont=sf,textfont=sf]{subfig}
\else
  \usepackage[caption=false,font=footnotesize]{subfig}
\fi
\usepackage{bm}
\mathchardef\mhyphen="2D 

\newtheorem{hyp}{Hypothesis}
\newtheorem{theorem}{Theorem}
\usepackage[switch]{lineno}
\usepackage{xcolor}
\usepackage{color,soul}

\usepackage{comment}

\begin{document}
%
\title{SS-GUMAP, SL-GUMAP, SSSL-GUMAP: \\
Fast UMAP Algorithms for Large Graph Drawing }
%
%
%
%

\author{Amyra Meidiana and Seok-Hee Hong
\IEEEcompsocitemizethanks{\IEEEcompsocthanksitem A. Meidiana and S.-H. Hong are with the University of Sydney.\protect\\
E-mail: \{amyra.meidiana,seokhee.hong\}@sydney.edu.au}
}

\IEEEtitleabstractindextext{%
\begin{abstract}
UMAP is a popular neighborhood-preserving dimension reduction (DR) algorithm. However, its application for graph drawing has not been evaluated. 
Moreover, a naive application of UMAP to graph drawing would include $O(nm)$ time all-pair shortest path computation, which is not scalable to visualizing large graphs.
%
In this paper, we present fast UMAP-based for graph drawing. Specifically, we present three fast UMAP-based algorithms for graph drawing:
(1) The SS-GUMAP algorithm utilizes spectral sparsification to compute a subgraph $G'$ preserving important properties of a graph $G$, reducing the $O(nm)$ component of the runtime to $O(n^2 \log n)$ runtime; 
(2) The SSL-GUMAP algorithm reduces the $k$NN ($k$-Nearest Neighbors) graph computation from $O(n \log n)$ time to linear time using partial BFS (Breadth First Search), and the cost optimization runtime from $O(n)$ time to sublinear time using edge sampling; 
(3) The  SSSL-GUMAP algorithm combines both approaches, for an overall $O(n)$ runtime.
Experiments demonstrate that SS-GUMAP runs 28\% faster than GUMAP, a naive application of UMAP to graph drawing, with similar quality metrics, while SL-GUMAP and SSSL-GUMAP run over 80\% faster than GUMAP with less than 15\% difference on average for all quality metrics.
We also present an evaluation of GUMAP to tsNET, a graph layout based on the popular DR algorithm t-SNE. GUMAP runs 90\% faster than tsNET with similar neighborhood preservation and, on average, 10\% better on quality metrics such as stress, edge crossing, and shape-based metrics, validating the effectiveness of UMAP for graph drawing.
\end{abstract}

}

\maketitle

\IEEEdisplaynontitleabstractindextext

%
\IEEEpeerreviewmaketitle

\IEEEraisesectionheading{\section{Introduction}\label{sec:introduction}}

%
%
%
%
\IEEEPARstart{T}{he} \emph{UMAP (Uniform Manifold Approximation and Projection)} algorithm~\cite{mcinnes2018umap} is a popular dimension reduction (DR) algorithm for visualizing high-dimensional data, being considered one of the top performing DR algorithms~\cite{espadoto2019toward}. It is based on the assumption that there exists a locally connected manifold where the high-dimensional data is uniformly distributed, and aims to preserve the topological structure of the manifold in a low-dimensional projection.

The projection is computed by first constructing a weighted \emph{k-nearest neighbor ($k$NN) graph} of the high-dimensional data, then computing a layout of the $k$NN graph that preserves its desired characteristics. 
The layout for the $k$NN graph is computed using a cost function aiming to minimize the disparity between the high-dimensional similarities of the data points and the low-dimensional similarity in the projection, with the cost minimization commonly done using gradient descent.

Overall, the runtime of UMAP is bounded by the runtime of computing the kNN graph, which can be computed in $O(n \log n)$ time~\cite{beygelzimer2006cover,liu2004investigation}, or empirically in $O(n^{1.14})$ time with approximate $k$NN~\cite{dong2011efficient}. 
The gradient descent iterations used to iteratively optimize the projection takes linear time in the number of edges of the kNN, which in turn would be $O(kn)$ time in the size of the high-dimensional data, where $k$ is the number of neighbors in the $k$NN and $n$ is the number of data points (i.e., linear in the input size $n$ if the neighborhood size $k$ is constant w.r.t. $n$).

With its superior performance in high-dimensional data visualization, it is a natural step to consider UMAP for graph drawing, such as has been done with other DR algorithms; a relatively recent example is the \emph{tsNET} algorithm ~\cite{kruiger2017graph} based on the popular \emph{t-SNE} algorithm~\cite{van2008visualizing}.

However, a naive application of UMAP to graph drawing would encounter scalability issues due to the need to compute the shortest path distance for all pairs of vertices, which takes $O(nm)$ time~\cite{chan2012all}, where $n = |V|$ and $m = |E|$; this is especially an issue for \emph{dense} graphs, where $m = O(n^2)$. Furthermore, recent work in graph drawing has presented have reduced the runtime of popular, traditionally $O(n^2)$ time algorithms to overall linear and even sub-linear time, including for force-directed~\cite{gove2019random,meidiana2020sublinear-tvcg} and stress minimization~\cite{ortmann2016sparse,meidiana2021stress} algorithms. Therefore, there is an opening for optimizing the runtime of UMAP specifically for graph drawing.



In this paper, we present fast UMAP-based algorithms, specifically designed for graph drawing. More specifically, we break down UMAP-based graph drawing algorithms into three components: computation of all-pairs shortest paths (denoted as \textbf{C0}, originally $O(nm)$ time), computation of the kNN (denoted as \textbf{C1}, originally $O(n \log n)$ time), and optimization of the cost function (denoted as \textbf{C2}, originally $O(n)$ time).

To reduce the runtime of the steps, our algorithms leverage graph sparsification and traversal methods to reduce the runtime of distance computations of graphs, where methods to speed up kNN computations on high-dimensional data, such as vantage-point trees~\cite{yianilos1993data}, cannot be applied directly, in order to speed up the cost optimization loop for graph drawing.

The SS-GUMAP algorithm newly integrates UMAP with \emph{spectral sparsification}~\cite{spielman2007spectral}, which computes a subgraph preserving important structural properties such as connectivity and commute distance, to reduce the runtime of shortest path computations; the SL-GUMAP algorithm further newly integrates partial BFS (Breadth-First Search) and edge sampling to achieve $O(n)$ runtime cost optimization for graph drawing with UMAP; and the SSSL-GUMAP algorithm combines both approaches for an $O(n)$ runtime cost optimization.

To put our work into context, we also present an evaluation of GUMAP, a naive application of UMAP for graph drawing, against tsNET, based on the popular DR algorithm t-SNE, which focuses on the same criteria of neighbourhood preservation as UMAP. Specifically, unlike existing comparisons of tsNET to other graph drawing algorithms that use random initialization~\cite{zhu2020drgraph,miller2023balancing}, we compare GUMAP to tsNET with Pivot MDS (PMDS) initialization, which has been shown to produce better quality than random initialization~\cite{kruiger2017graph}, for a fairer comparison of both algorithms.


The contributions of this paper can be summarized as:

\begin{enumerate}
    
    \item We present SS-GUMAP, which utilizes spectral sparsification to compute a subgraph $G'$ of a graph $G$ with $O(n \log n)$ edges, preserving important structural properties such as connectivity and commute distance of $G$. We then run UMAP on the subgraph $G'$, reducing the $O(nm)$ time shortest path computation of $G$ for C0 to $O(n^2 \log n)$ time.
    
     \item We present SL-GUMAP, which utilizes partial BFS to compute the shortest paths and kNN for C0 and C1 $O(n)$ time, and sliding window-based edge sampling to reduce the cost optimization iterations for C2 to sub-linear time. 
        
     \item We present SSSL-GUMAP, which combines both approaches by running SL-GUMAP on the spectral sparsification $G'$ of a graph $G$ to reduce C0 and C1 to $O(n)$ time and C2 to sub-linear time, with faster runtime than SL-GUMAP in execution.

    
    \item We implement and evaluate our algorithms against GUMAP, a naive application of UMAP to graph drawing, with various real-world and synthetic graphs. 
    Experiments demonstrate that our algorithms are significantly faster than UMAP, at 28\% faster by  SS-GUMAP, and over 80\% faster by SL-GUMAP and SSSL-GUMAP, while maintaining good quality metrics (neighborhood preservation, stress, edge crossing, shape-based metrics).
     
    \item As a secondary experiment, we evaluate GUMAP against tsNET~\cite{kruiger2017graph}, a graph-drawing algorithm based on the pt-SNE algorithm~\cite{van2008visualizing}. Different from the existing comparison~\cite{miller2023balancing}, we use tsNET with PMDS initialization for a fairer comparison. Experiments show that GUMAP runs on average over 90\% faster than tsNET with similar neighborhood preservation, and better performance on other quality metrics (stress, edge crossing, shape-based metrics).

\end{enumerate}


Table \ref{tab:runtime_umap} summarizes   
the theoretical runtime analysis of our fast algorithms against GUMAP, the naive application of UMAP for graph drawing, per component, where the bold entries highlight our contribution in the runtime reduction. 

\begin{table}[h]
    \centering
    \scriptsize
    \caption{Runtimes of fast \texttt{GUMAP} algorithms}
    \begin{tabular}{|c|c|c|c|c|}
        \hline 
        Algo. & C0 & C1 & C2 & Overall \\ \hline
        GUMAP & $O(nm)$ & $O(n \log n)$ & $O(n)$ & $O(nm)$ \\ \hline
            SS-GUMAP & $\pmb{O(n^2 \log n)}$ & $O(n \log n)$ & $O(n)$ & $\pmb{O(n^2 \log n)}$ \\ 
            & & & & (Thm. \ref{theorem:ssgumap}) \\ \hline
        SL-GUMAP & $\pmb{O(n)}$ & $\pmb{O(n)}$ & $\pmb{O(n^{0.9})}$ & $\pmb{O(n)}$ (Thm. \ref{theorem:slgumap}) \\ \hline
        SSSL-GUMAP & $\pmb{O(n)}$  & $\pmb{O(n)}$  & $\pmb{O(n^{0.9})}$  & $\pmb{O(n)}$ (Thm. \ref{theorem:ssslgumap})  \\ \hline
    \end{tabular}
    \label{tab:runtime_umap}
\end{table}

\section{Related Work}
\label{sec:related_work}

\subsection{UMAP}
\label{sec:relwork_umap}

UMAP (Uniform Manifold Approximation and Projection)~\cite{mcinnes2018umap} is a manifold learning technique for DR. 
It assumes that there exists a locally connected manifold where the high-dimensional data is uniformly distributed, and aims to compute a low-dimensional projection preserving the topological structure. 
UMAP consists of two major steps:
\begin{description}
    \item[\textit{C1:}] Construct the $k$NN ($k$-Nearest Neighbor) graph of the input data that represents the manifold.
    \item[\textit{C2:}] Compute a layout for the $k$NN graph preserving desired characteristics.
\end{description}

%
C1 computes a weighted $k$NN graph, where each edge between a data point $v_i$ and one of its nearest neighbors $v_j$ is weighted using the \emph{high-dimensional similarity} computed as $h_{ij} = exp[(-d_{ij}-\rho_i)/\sigma_i]$, where 
$d_{ij}$ is the high-dimensional distance between $v_i$ and $v_j$, $\rho_i$ is the distance from $v_i$ to its nearest neighbor, and $\sigma_i$ is a bandwidth similar to the perplexity used in t-SNE~\cite{van2008visualizing}. 

C2 computes a layout for the $k$NN graph by minimizing the difference between the aforementioned high-dimensional similarity $h_{ij}$ between data points $v_i$ and $v_j$, and the \emph{low-dimensional similarity} between $X_i$ and $X_j$, the low-dimensional projections of $v_i$ and $v_j$, which is computed as $w_{ij} = \left (1 + \alpha || X_i - X_j ||^{2\beta}_2\right )$ where $\alpha,\beta$ are positive constants.
The projection is computed by minimizing the \emph{cost function} $C$ through gradient descent~\cite{mcinnes2018umap}: 

\[
 C = \sum_{i\neq j} h_{ij} \log \left ( \frac{h_{ij}}{w_{ij}} \right ) + (1 - h_{ij}) \log \left ( \frac{1 - h_{ij}}{1 - w_{ij}} \right )   
\]

Overall, the time complexity of UMAP is bounded by the computation of the $k$NN graph. 
If the distances between data points are not precomputed, building the $k$NN graph can be done in $O(dn \log n)$ time, where $d$ is the number of dimensions, by computing a $kd$ tree of the data points~\cite{beygelzimer2006cover}.
Approximate $k$NN can also be used, which has been shown to empirically run in  $O(n^{1.14})$ time~\cite{mcinnes2018umap,dong2011efficient}.

Compared to other state-of-the-art DR algorithms such as t-SNE, UMAP runs much faster, with its main strength being neighborhood preservation, especially in displaying clustering structures~\cite{mcinnes2018umap,espadoto2019toward}. However, this may come at the expense of preservation of global structure, such as distance preservation~\cite{mcinnes2018umap,espadoto2019toward}. 
Other UMAP-based algorithms attempt to address this limitation, such as UMATO~\cite{jeon2022uniform}, which uses \emph{extended} nearest neighbors to compute the global skeleton structure before optimizing the local neighborhoods.

\subsection{Graph Sampling and Spectral Sparsification}


A popular approach for drawing large graphs uses \emph{graph sampling}, where a significantly smaller subgraph $G'$ is sampled and then drawn instead of the original graph $G$~\cite{leskovec2006sampling}. 
While the sampled subgraph $G'$ can be drawn much faster, the challenge is to compute a good sample that is representative of the original graph $G$.
For example, the simplest stochastic sampling methods, such as Random Vertex or Random Edge sampling, often fail to preserve important structural properties such as connectivity~\cite{wu2017evaluation,hong2018bc}.

Recent methods to improve the quality of sample graphs for visualization utilize the topological structure of the graph $G$ in sampling, such as always including cut vertices or separation pairs~\cite{hong2018bc,meidiana2019topology}, as well as using spectral sparsification~\cite{eades2017drawing,JM,meidiana2019topology}.

{\em Spectral sparsification (SS)} computes a subgraph $G'$, where the Laplacian quadratic form is approximately the same as the original graph $G$ on all real vector inputs.
Every \(n\)-vertex graph has a spectral approximation with \(O(n \log n)\) edges~\cite{spielman2011spectral}, which can be computed based on the {\em effective resistance} values of edges, closely related to important properties such as connectivity, commute distance, and clustering~\cite{spielman2011graph} and can be computed in near-linear $O(n \log n)$ time~\cite{spielman2011graph}.

Spectral sparsification has been shown to be effective for sampling large graphs, outperforming random sampling methods on sampling quality metrics~\cite{eades2017drawing,JM,JMjournal,meidiana2019topology}. Moreover, SS has been integrated with graph topology, such as biconnected component decomposition using BC (Block-Cut vertex) trees~\cite{Hu3} or triconnected component decomposition using SPQR trees~\cite{meidiana2019topology} to reduce the runtime while improving the sampling quality.

SS has also been successfully used for state-of-the-art fast graph drawing algorithms for large and complex graphs. 
For example, the {\em SublinearForce} framework~\cite{meidiana2020sublinear-tvcg} utilizes $SS$ to significantly improve the runtime of the force computation of the force-directed algorithm, from quadratic to sublinear time, also resulting in better quality drawings. 
Similarly, $SS$ has also been used for stress-based algorithms {\em SublinearSM} and {\em SublinearSGD}~\cite {meidiana2021stress}, significantly improving the runtime of stress computation, from quadratic to sublinear time, while maintaining similar quality drawings.

\section{Fast UMAP-Based Graph Drawing algorithms}
\label{sec:umap_algs}

\subsection{GUMAP}

Before presenting our fast UMAP algorithms, we first describe GUMAP, a naive application of UMAP to graph drawing. GUMAP consists of the following steps:

\begin{description}
    \item[\textbf{GUMAP}]
    \item[\textit{C0:}] Compute the distance matrix of the graph $G$.
    \item[\textit{C1:}] Compute the $k$NN graph $G_k$. 
    \item[\textit{C2:}] Compute a layout for $G_k$. 
\end{description}

To apply UMAP for graph drawing, we add the shortest path computation in C0, which is commonly required for DR-based graph drawing algorithms, on top of C1 and C2, which originally came from UMAP (see Section \ref{sec:relwork_umap} for details).

\begin{algorithm}
 \caption{\textbf{GUMAP}}
 \begin{algorithmic}[1]
 \State \textbf{Input:} Graph $G = (V,E)$, neighborhood size $k$, \# of iterations $t$

 \State // C0
 \State $S=$AllPairsShortestPath($G$) \label{line:allpairs}

 \State // C1
 \State $G_k = kNN(S)$ \label{line:knn}

 \State // C2
 \State $X:$ Coordinates of the vertices in the drawing \label{line:cstart}
 \State Initialize the coordinates in $X$ using spectral layout
  \For{iterations in 1 to $t$}
    \For{$(v_a,v_b) \in E$}
        \State $X_a,X_b$: Coordinates of vertices $v_a,v_b$
        \State $w_{ab} = \left (1 + \alpha || X_a - X_b ||^{2\beta}_2\right )$
        \State Compute gradient of $C$ at $a,b$ based on $w_{ab}$
        \State Move $X_a,X_b$ based on the gradient
    \EndFor
 \EndFor \label{line:cend}
 \State $D$: layout of $G$ with $X$ as coordinates
 \State \textbf{return} $D$
 \end{algorithmic}
 \label{alg:umapg}
\end{algorithm}

Algorithm \ref{alg:umapg} shows a high-level description of GUMAP. Line \ref{line:allpairs} corresponds to C0, to compute the distance matrix by computing the shortest path length between every pair of vertices. Line \ref{line:knn} corresponds to C1, to compute the $k$NN graph based on the distance matrix from C0. Lines \ref{line:cstart} to \ref{line:cend} correspond to C2, comprised of the gradient descent procedure used in UMAP to minimize the cost function $C$.

\begin{theorem}
   GUMAP runs in $O(nm)$ time (C0: $O(nm)$ time, C1: $O(n \log n)$ time, C2: $O(n)$ time) with $O(n \log n + m)$ space complexity.
\end{theorem}

\begin{proof}
The runtime of C0 is $O(nm)$ time, for computing all-pairs shortest paths. The runtime of C1, for computing $k$NN graph, can be done in $O(n \log n)$ time by using $kd$ trees, where $k$ and $d$ are constant with respect to $n$. 
C2 takes linear time in the number of edges of the $k$NN graph; taking $k$ as $O(1)$ in the number of vertices, this is equivalent to $O(n)$ time as the number of edges in the $k$NN graph is of order $O(kn)$. 

Meanwhile, the space complexity consists of the $O(n + m)$ space needed to store the input data, the $O(n \log n)$ space to compute the kNN, and the $O(n)$ space to store the $v_{ij}$ values.

\end{proof}

Note that there is an inherent $O(n + m)$-time component for reading the input and rendering the output of graph drawing, due to the input data being $O(n + m)$ in size. However, as this component is shared throughout every algorithm, we have omitted this component in our runtime analysis.


\subsection{SS-GUMAP: UMAP with Spectral Sparsification}

As the first step to reduce the runtime of GUMAP, we first present SS-GUMAP, which runs GUMAP as-is on the spectral sparsification $G'$ of a graph $G$. 
As the runtime of shortest path computation for graphs is dependent on the number of edges ($O(nm)$ time), we use spectral sparsification as pre-processing to reduce the runtime of the distance computation, while preserving the graph's ground truth structure. Our usage of the near-linear-time spectral sparsification is in the same vein as the usage of near-linear-time pre-processing in other fast DR- and stress-based graph-drawing algorithms, such as Pivot MDS~\cite{brandes2006eigensolver} and Sparse Stress Minimization~\cite{ortmann2016sparse}, which use near-linear-time pivot computation as pre-processing to obtain a linear-time cost minimization step.

\begin{algorithm}
 \caption{\textbf{SS-GUMAP}}
 \begin{algorithmic}[1]
 \State \textbf{Input:} Graph $G=(V,E)$, spectral sparsification $G' = (V,E')$ of $G$, neighborhood size $k$, \# of iterations $t$

 \State // C0
 \State $S'=$AllPairsShortestPath($G'$)

 \State // C1
 \State $G'_k = kNN(S')$

 \State // C2
 \State $X:$ Coordinates of the vertices in the drawing
 \State Initialize the coordinates in $X$ using spectral layout
  \For{iterations in 1 to $t$}
    \For{$(v_a,v_b) \in E'$}
        \State $X_a,X_b$: Coordinates of vertices $v_a,v_b$
        \State $w_{ab} = \left (1 + \alpha || X_a - X_b ||^{2\beta}_2\right )$
        \State Compute gradient at $a,b$ based on $w_{ab}$
        \State Move $X_a,X_b$ based on the gradient
    \EndFor
 \EndFor
 \State $D'$: layout of $G'$ with $X$ as coordinates
 \State $D$: $D'$ with all edges of $G$ added in
 \State \textbf{return} $D$
 \end{algorithmic}
 \label{alg:umapss}
\end{algorithm}

Algorithm \ref{alg:umapss} shows a high-level description of SS-GUMAP. 
As pre-processing, we compute a sparsification $G' = (V, E' \subset E)$ of a graph $G = (V, E)$ using spectral sparsification: first, compute the effective resistance values of all edges in $E$, which can be done in $O(m \log n)$ time. $G'$ is then computed by computing the maximum spanning tree with effective resistance values as weights, then adding edges in decreasing order of effective resistance values until $n \log n$ edges are included. 
In the case that a graph is disconnected, an arbitrarily large ``infinite'' value can be assigned to pairs of vertices with no paths connecting them.

We use the spectral layout~\cite{koren2003spectral} for initialization, based on the usage of spectral layout for initialization in UMAP for DR, which has been shown to converge faster and stable, compared to random initialization~\cite{mcinnes2018umap}. The spectral layout for graph drawing takes linear time in the number of edges, due to the Laplacian matrix only having non-zero entries in the diagonal and on entries corresponding to edges, and runs very fast in practice~\cite{koren2003spectral}.

SS-GUMAP runs steps C0-2 on $G'$, to obtain the coordinates $X$ of the vertices, and consequently returns a drawing $D'$ of $G'$; we use the same cost optimization steps as the existing base UMAP algorithm. 
To obtain a drawing $D$ of the original graph $G$, add the sparsified edges into $D'$.

\begin{theorem} \label{theorem:ssgumap}
    SS-UMAP runs in $O(n^2 \log n)$ time (C0: $O(n^2 \log n)$ time, C1: $O(n \log n)$ time, C2: $O(n)$ time) with $O(n \log n + m)$ space complexity.
\end{theorem}

\begin{proof}
By sparsifying the edge set to $O(n \log n)$ size, all-pairs shortest path computation for $C0$ can be reduced from $O(nm)$ time to $O(n^2 \log n)$ time, with C1 and C2 maintaining their $O(n \log n)$ and $O(n)$ runtimes respectively.
On top of the $O(n + m)$ space required for the input graph, all-pairs shortest path takes linear space in the number of edges, which is $O(n \log n)$ after sparsification, and storing the $v_{ij}$ values take $O(n)$ space, for a total of $O(n \log n)$ space.
\end{proof}

\subsection{SL-GUMAP: Linear-Time UMAP for Graph Drawing}

In this section, we present SL-GUMAP, a linear-time UMAP-based graph drawing algorithm. 
Specifically, the runtime of Steps 0 and 1 is reduced from $O(n \log n)$ time to $O(n)$ time by using partial BFS, and the runtime of Step 2 is reduced from linear time to sublinear time by using a sampling approach in each iteration.

Algorithm \ref{alg:lumap} shows a high-level overview of the SL-GUMAP algorithm,
which consists of two parts: part 1, computing the partial distance matrix and $k$NN graph, corresponding to C0 and C1; and part 2, the optimization loop of the cost function, corresponding to C2. 
Most of the work is contained in the subroutines PartialBFS for part 1 and GradSample for part 2, which we present in the following sections.

\begin{algorithm}
 \caption{\textbf{SL-GUMAP}}
 \begin{algorithmic}[1]
 \State \textbf{Input:} Graph $G = (V,E)$, neighborhood size $k$, sample size $samp$, \# of iterations $t$
 \State // C0 and C1
 \State $S,G_k$: PartialBFS($G$, $k$) // partial distance matrix and $k$NN graph

 \State // C2
 \State $X:$ Coordinates of the vertices in the drawing
 \State Initialize the coordinates in $X$ using spectral layout
 \State GradSample($G_k$, $X$, $samp$, $t$)
 \State $D$: Drawing of $G$ with $X$ as coordinates
 \State  \textbf{return} $D$
 \end{algorithmic}
 \label{alg:lumap}
\end{algorithm}

\subsubsection{Partial BFS for $k$NN Graph Computation}

For high-dimensional data, computing the $k$NN graph forms the bottleneck of the runtime of UMAP, taking $O(n \log n)$ time. 
When applying UMAP for drawing a graph $G = (V, E)$, an additional step is required to compute the distance matrix of the graph, as the distances between vertices cannot be computed the same way as the distance between two high-dimensional data points. 
Computing the full distance matrix of a graph takes $O(nm)$ time~\cite{chan2012all}, which can be prohibitively slow, especially for dense graphs ($m = O(n^2)$). 
However, as UMAP only needs the $k$NN graph to compute the layout, only the distances from each vertex to its $k$ nearest neighbors need to be computed.

\begin{algorithm} [t!]
 \caption{\textbf{PartialBFS}}
 \begin{algorithmic}[1]
 \State Graph $G = (V,E)$, neighborhood size $k$
 \State $S$: $n \times n$ matrix initialized to infinity // distance matrix
 \State $E_k =$ [] // edge list for $k$NN graph
 \For{$v \in V$}
    \State $visited=[]$ // list of already visited vertices
    \State $tovisit=[v]$ // list of vertices to be visited
    \State $dist=$1
    \While{$visited.size \leq k$} \label{line:traversestart}
        \State $newvisit=[]$
        \For{$u \in tovisit$}
            \State $visited.push(u)$
            \State $S[v][u]=dist$
            \State $E_k.push((v,u))$
            \If{$visited.size > k$} \label{line:kcheckstart}
                \State break
            \EndIf \label{line:kcheckend}
            \For{$w$ where $(u,w) \in E$}
                \If{$w \notin visited$ and $w \notin tovisit$}
                    \State $newvisit.push(w)$
                \EndIf
            \EndFor
        \EndFor
        \State $tovisit = newvisit$
        \State $dist++$
    \EndWhile \label{line:traverseend}
 \EndFor
 \State $G_k = (V, E_k)$ // $k$NN graph
 \State \textbf{return} $S,G_k$
 \end{algorithmic}
 \label{alg:pbfs}
\end{algorithm}

To reduce the runtime of this step, we use partial \emph{BFS (Breadth-First Search)}: 
for each vertex $v \in V$, perform a BFS starting from $v$ and record the distance from $v$ to the visited vertices, stopping once $k$ other vertices have been visited; this is done to compute the $k$-nearest neighbors to reduce the runtime of shortest path computations. When multiple vertices have the same shortest distance to $v$, ties are broken at random.

As a preliminary experiment, we compared the partial BFS approach to a pivot-based method such as used in Pivot MDS~\cite{brandes2006eigensolver}, where a partial distance matrix is constructed by computing the distances only between each vertex and vertices in a pivot set $P \subset V, |P| = O(1)$ w.r.t. $n$. However, the pivot-based approach resulted in far lower neighborhood preservation, due to the pivot-based distance matrix leading to the pivots forming the $k$NN of every vertex. Meanwhile, partial BFS is able to maintain a higher level of neighborhood preservation; therefore, we focus on the partial BFS approach. This method also differs from the partial distance matrix computation, such as in DRGraph~\cite{zhu2020drgraph}, by computing the partial BFS until $k$ vertices are visited, rather than limiting to direct neighbors.

In addition to reducing the runtime of C0, computing the partial BFS also eliminates the need to run C1, as the $k$ vertices visited in the partial BFS starting from a vertex $v$ already form the $k$-nearest neighbors of $v$, 
i.e., the $k$NN graph is already computed. 
Therefore, this reduces the runtime of both Step 0 and Step 1 to $O(kn)$ time in practice; i.e., $O(n)$ if $k = O(1)$ w.r.t. $n$.

Algorithm \ref{alg:pbfs} shows the steps of PartialBFS in detail. 
The steps are roughly similar to a regular BFS traversal, however, with an additional check to stop the traversal when $k$ other vertices have been visited (lines \ref{line:kcheckstart}-\ref{line:kcheckend}). 
In addition, the distance between the starting vertex $v$ and each vertex visited in the traversal starting from $v$ is recorded, using a counter that tracks the current BFS level. 

The returned distance matrix  $S$ is not a full all-pairs distance matrix: for each row corresponding to a vertex $v$, only $k$ columns corresponding to the $k$ nearest neighbors will have non-infinite distances, as the matrix is initialized to infinity and only the entries corresponding to the $k$ nearest neighbors are updated in the traversal loop (lines \ref{line:traversestart}-\ref{line:traverseend}).

\subsubsection{Sublinear-Time Cost Optimization}

We now present how SL-GUMAP reduces the runtime of C2 to sublinear in the number of vertices, by sampling a sublinear number of edges in each iteration.

C2 involves iteratively minimizing the cost function $C$ of UMAP, using gradient descent. 
This requires computing the low-dimensional similarity $w_{ij}$ on each edge $(i,j)$ of the $k$NN graph. 
As such, the time complexity is $O(kn)$ time, i.e., linear in the number of edges of the $k$NN graph, which is linear in the number of vertices if $k = O(1)$ w.r.t $n$. 
Nevertheless, there is still an opening to further reduce the runtime to sublinear in the number of vertices.

\begin{algorithm}[t!]
 \caption{\textbf{GradSample}}
 \begin{algorithmic}[1]
 \State \textbf{Input:} $k$NN graph $G_k = (V,E_k)$, coordinates of vertices $X$, sample size exponent $samp$, \# of iterations $t$
 \State Shuffle contents of $E_k$
 \State $j=0$
 \State $s = |E_k|^{samp}$
 \For{iterations in 1 to $t$}
    \For{$i \in [0,s]$} \label{line:windowstart}
        \State $(a,b) = E_k[(j+i) \mod |E_k|]$
        \State $X_a,X_b$: Coordinates of vertices $a,b$
        \State $w_{ab} = \left (1 + \alpha || X_a - X_b ||^{2\beta}_2\right )$
        \State Compute gradient at $a,b$ based on $w_{ab}$
        \State Move $X_a,X_b$ based on the gradient
    \EndFor \label{line:windowend}
    \State $j = (j+s) \mod |E_k|$
 \EndFor
 \end{algorithmic}
 \label{alg:sublc}
\end{algorithm}

In SL-GUMAP, we reduce the runtime of the gradient descent loop in Step 2 by using \emph{random sampling} of the edges of the $k$NN graph. 
To simulate randomness while keeping the empirical runtime low, we sample the edges of the $k$NN graph using a \emph{sliding window} technique,
used in linear- and sublinear-time force-directed graph drawing algorithms~\cite{gove2019random,meidiana2020sublinear}.

Algorithm \ref{alg:sublc} presents the details of the GradSample subroutine. 
Initially, the array of edges of the $k$NN graph $E_k$ is shuffled to introduce randomness. 
In each iteration, a sliding window of size $n^{0.9}$ is defined, starting at an index $i_{start}$. The exponent 0.9 has been chosen based on preliminary experiments that show a window size of $O(n^{0.9})$ provided a good runtime-quality trade-off compared to lower exponent sizes. 
Instead of computing $w_{ij}$ and subsequently the gradient over all edges in the $k$NN, we only compute the gradient for edges with the index $(i_{start} + l) \% n$ for $l = 1, ..., floor(n^{0.9})$ (lines \ref{line:windowstart}-\ref{line:windowend}). 
After the iteration is finished, set the new starting index as $(i_{start} + floor(n^{0.9})) \% n$. This process ensures that all edges in $E_k$ are eventually considered in the gradient descent computation, while keeping the runtime of each individual iteration low.

\begin{theorem} \label{theorem:slgumap}
    SL-GUMAP runs in $O(n)$ time (C0: $O(n)$ time, C1: $O(n)$ time, C2: $O(n^{0.9})$ time) with $O(n + m)$ space complexity.
    
\end{theorem}

\begin{proof}
The partial BFS runs in $O(n)$ time, as for each vertex, the traversal is stopped after a constant number of other vertices are visited, reducing C0 and C1 to $O(n)$ time. GradSample for C2 runs in $O(n^{0.9})$ time, due to each iteration selecting $O(n^{0.9})$ vertices to compute the gradient on. With the partial BFS, shortest path computation only need $O(n)$ space, with $v_{ij}$ values for the gradient descent step taking $O(n)$ space, for a total of $O(n + m)$ space complexity when combined with the $O(n + m)$ space required for the input graph.

\end{proof}

\subsection{SSSL-GUMAP}

Finally, we present SSSL-GUMAP, which combines the approaches of SS-GUMAP and SL-GUMAP. 
As with SS-GUMAP, as pre-processing, we compute the spectral sparsification $G'$ of $G$. 
We then use $G'$ as the input for SL-GUMAP to compute the coordinates $X$ of the vertices in $V$ in the drawing. 
Finally, as with SS-GUMAP, we add all the vertices in $E$ to the drawing to obtain a drawing of the original graph $G$. 
Algorithm \ref{alg:lumapss} shows a high-level overview of SSSL-GUMAP.

\begin{theorem} \label{theorem:ssslgumap}
    SSSL-GUMAP runs in $O(n)$ time (C0: $O(n)$ time, C1: $O(n)$ time, C2: $O(n^{0.9})$ time) with $O(n + m)$ space complexity.
   
\end{theorem}

\begin{proof}

    The partial BFS for C0 and C1 runs in $O(n)$ time, the same as SL-GUMAP. GradSample for C2 runs in $O(n^{0.9})$ time, due to the sampling done in each iteration. As with SL-GUMAP, the partial BFS takes $O(n)$ space and $v_{ij}$ values for the gradient descent step tasks $O(n)$ space, for a total of $O(n + m)$ space, including the $O(n + m)$ space required for the input graph.
    
\end{proof}

\begin{algorithm}
 \caption{\textbf{SSSL-GUMAP}}
 \begin{algorithmic}[1]
 \State \textbf{Input:} Graph $G=(V,E)$, spectral sparsification $G' = (V,E')$ of $G$, neighborhood size $k$, sample size $samp$, \# of iterations $t$

 \State // C0 and C1
 \State $S,G_k$: PartialBFS($G$, $k$) // partial distance matrix and $k$NN graph

 \State // C2
 \State $X:$ Coordinates of the vertices in the drawing
 \State Initialise the coordinates in $X$ using spectral layout
 \State GradSample($G_k$, $X$, $samp$, $t$)

 \State $D'$: Drawing of $G'$ with $X$ as coordinates
 \State $D$: $D'$ with all edges in $G$ added in
 
 \State  \textbf{return} $D$
 \end{algorithmic}
 \label{alg:lumapss}
\end{algorithm}

By combining both the spectral sparsification approach of SS-GUMAP and the partial BFS and sublinear-time cost optimization computation from SL-GUMAP, we expect SSSL-GUMAP to run even faster in practice, while maintaining a reasonable level of quality.

\section{Comparison Experiments: SS-GUMAP, SL-GUMAP, SSSL-GUMAP vs GUMAP}
\label{sec:eval}

\subsection{Experiment Design}

After validating the effectiveness of utilizing UMAP for graph drawing, we now present experiments evaluating the effectiveness and efficiency of our fast UMAP-based graph drawing algorithms. We perform three experiments comparing our algorithms, SS-GUMAP, SL-GUMAP, and SSSL-GUMAP, to the baseline GUMAP. We aim to demonstrate that our algorithms run faster than the baseline GUMAP, without compromising quality with unfavorable runtime-quality trade-offs.  

\subsubsection{Implementation and Parameters}

We implement our algorithms in Python, based on the original implementation of McInnes et al.~\cite {mcinnes2018umap}, combined with NetworkX~\cite{hagberg2008exploring} for shortest path computation. For spectral sparsification of SS-GUMAP and SSSL-GUMAP, we use the implementation of GSparse~\cite{gsparse}. We use the spectral layout for the initialization, as commonly used with UMAP for DR.

For all algorithms, we use $d=2$ (i.e., two-dimensional drawing) and $k=15$ for the $k$NN graph neighbourhood size based on the default parameters, i.e., both $d$ and $k$ are always $O(1)$ with respect to $n$. For each graph, we run each algorithm five times and take the average runtime and quality metrics over all five runs. The algorithms are run on a Linux box with Intel CORE i7  and  16GB RAM.

We expect that all of our algorithms will run faster than GUMAP. 
Specifically, among our algorithms, we expect SSSL-GUMAP to run the fastest, followed by SL-GUMAP and SS-GUMAP. 
Quality metric-wise, we expect SS-GUMAP to obtain the highest metrics, followed by SL-GUMAP and SSSL-GUMAP.

\begin{table}[h]
    \centering
    \caption{Data sets for experiments} 
    \subfloat[Benchmark scale-free graphs]{
    \begin{tabular}{|l|c|c|c|}
    \hline
        graph & $|V|$ & $|E|$ & density \\ \hline
        G\_13\_0 & 1647 & 6487  & 3.94 \\ \hline
        soc\_h & 2000 & 16097  & 8.05 \\ \hline
        Block\_2000 & 2000 & 9912  & 4.96 \\ \hline
        G\_4\_0 & 2075 & 4769 & 2.30 \\ \hline
        oflights & 2905 & 15645 & 5.39 \\ \hline
        tvcg & 3213 & 10140  & 3.16 \\ \hline
        facebook & 4039 & 88234  & 21.8 \\ \hline
        CA-GrQc & 4158 & 13422  & 3.23 \\ \hline
        EVA & 4475 & 4652 & 1.04 \\ \hline
        us\_powergrid & 4941 & 6594  & 1.33 \\ \hline
        as19990606 & 5188 & 9930  & 1.91 \\ \hline
        migrations & 6025 & 9378  & 1.57 \\ \hline
        lastfm\_asia & 7624 & 27906  & 3.65\\ \hline
        CA-HepTh & 8638 & 24827  & 2.87 \\ \hline
        CA-HepPh & 11204 & 117649  & 10.5 \\ \hline
    \end{tabular}
    }
    \\
    \subfloat[GION graphs]{
    \begin{tabular}{|l|c|c|c|}
    \hline
        graph & $|V|$ & $|E|$ &  density \\ \hline
        GION\_2 & 1159 & 6424 &  5.54 \\ \hline
        GION\_5 & 1748 & 13957 &  7.98 \\ \hline
        GION\_6 & 1785 & 20459 & 11.5 \\ \hline
        GION\_7 & 3010 & 41757 &  13.9 \\ \hline
        GION\_8 & 4924 & 52502 & 10.7 \\ \hline
        GION\_1 & 5452 & 118404 &  21.7 \\ \hline
        GION\_4 & 5953 & 186279 & 31.3\\ \hline
        GION\_3 & 7885 & 427406 & 54.2\\ \hline
    \end{tabular}
    }
    \\
    \subfloat[Mesh graphs]{
    \begin{tabular}{|l|c|c|c|}
    \hline
        graph & $|V|$ & $|E|$ & density \\ \hline
        cage8 & 1015 & 4994 & 4.92 \\ \hline
        bcsstk09 & 1083 & 8677 &  8.01 \\ \hline
        nasa1824 & 1824 & 18692 &  10.2 \\ \hline
        plat1919 & 1919 & 15240 & 7.94 \\ \hline
        sierpinski3d & 2050 & 6144 &3.00 \\ \hline
        3elt & 4720 & 13722 & 2.91 \\ \hline
        crack & 10240 & 30380 &  2.97 \\ \hline
    \end{tabular}
    }
    \label{table:data}
   \vspace{-2mm}
\end{table}

\subsubsection{Data sets}

We use a wide selection of graphs taken from data sets commonly used to evaluate graph drawing algorithms, including DR-based graph drawing algorithms~\cite{brandes2006eigensolver,kruiger2017graph}; we use the largest connected component of each graph. 
We divide the data used into three sets, to demonstrate the wide applicability of the algorithms: (1) benchmark \emph{real-world scale-free} graphs, with a globally sparse and locally dense structure~\cite{snapnets}; (2) \emph{GION} graphs, taken from biochemical RNA networks with large diameters~\cite{marner2014gion}; and (3) \emph{mesh} graphs with regular grid-like structures~\cite{davis2011university}; see Table \ref{table:data}.

\subsubsection{Quality Metrics}

We compare GUMAP and tsNET on runtime and four graph drawing quality metrics widely used in graph drawing literature: neighborhood preservation, edge crossing, shape-based metrics, and stress.
Edge crossing is one of the most popular {\em readability} metrics for graph drawing~\cite{battista1998graph}.

\emph{Neighborhood preservation}, often used for evaluating DR algorithms~\cite{espadoto2019toward}, measures how well the neighborhood $N_G(v, r)$, i.e., the set of vertices in $G = (V,E)$ with a shortest path distance at most $r$ from a vertex $v \in V$ (we set $r=2$, same as~\cite{kruiger2017graph}) 
is represented by the \emph{geometric} neighborhood $N_D(v, r)$ in a drawing $D$ of $G$ 
(i.e., the $|N_G(v, r)|$-closest neighbors of $v$ in $D$)~\cite{martins2015explaining}, computed as: 
$
    \sum_{v \in V} \frac{|N_G(v,r) \cap N_D(v,r)|}{|N_G(v,r) \cup N_D(v,r)|}.
$

\emph{Shape-based metrics}~\cite{eades2017shape} are designed to evaluate the quality of large graph drawing. They measure how faithfully the ``shape'' of the drawing, computed using the \emph{proximity graph}~\cite{toussaint1986computational} $P$ of the vertex point set, represents the ground truth structure of the graph $G$. 
The metric is computed using the Jaccard similarity of $G$ and $P$.
We use $dRNG$ (Degree-sensitive Relative Neighborhood Graph)~\cite{hong2022dgg} as the proximity graph, which computes more accurate shape-based metrics than the RNG proximity graph~\cite{eades2017shape}.

{\em Stress}, which measures how proportionally the geometric distances in the drawing represent the ground truth shortest path distances between vertices, is a widely used optimization criterion in both graph drawing~\cite{battista1998graph} and dimension reduction~\cite{espadoto2019toward}. 
Specifically, we use the aggregated stress formula~\cite{kruiger2017graph}:
$
 \frac{1}{n^2-n}\sum_{i, j \in |V|} \left ( \frac{d(v_i,v_j) - ||X_i - X_j||}{(d(v_i,v_j)} \right )^2.  $

For neighborhood preservation and shape-based metrics, a higher value is better, while for edge crossing and stress, a lower value is better.

We compute the {\em runtime improvement}  between GUMAP and our algorithms as the absolute difference between the two divided by the value for GUMAP.
For example, the runtime improvement by SL-GUMAP over GUMAP is computed as $\frac{time(G\mhyphen UMAP)-time(SL\mhyphen G\mhyphen UMAP)}{time(G\mhyphen UMAP)}$.
We define the difference in metrics where a lower value is better (stress, edge crossing) similarly.
For metrics where a higher value is better (neighborhood preservation, shape-based), we reverse the numerator; for example, for neighborhood preservation of SL-GUMAP compared to UMAP, we use the formula $\frac{np(SL\mhyphen G\mhyphen UMAP)-np(G\mhyphen UMAP)}{np(G\mhyphen UMAP)}$.

\begin{figure}[t]
    \centering
    \subfloat[Scale-free]{
        \includegraphics[width=0.99\columnwidth]{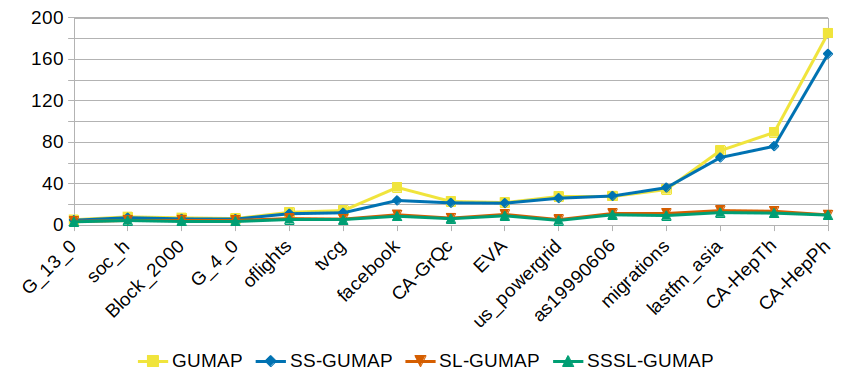}
        \label{fig:runtime_bench}
    }
    \qquad
    \subfloat[GION]{
        \includegraphics[width=0.9\columnwidth]{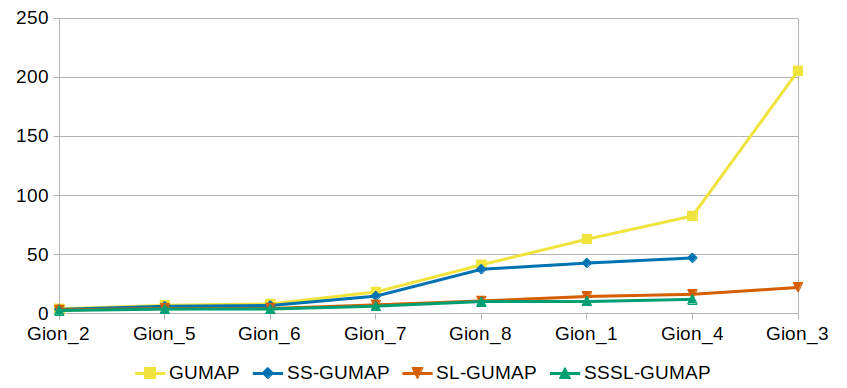}
        \label{fig:runtime_gion}
    }
    \qquad
    \subfloat[Mesh]{
        \includegraphics[width=0.9\columnwidth]{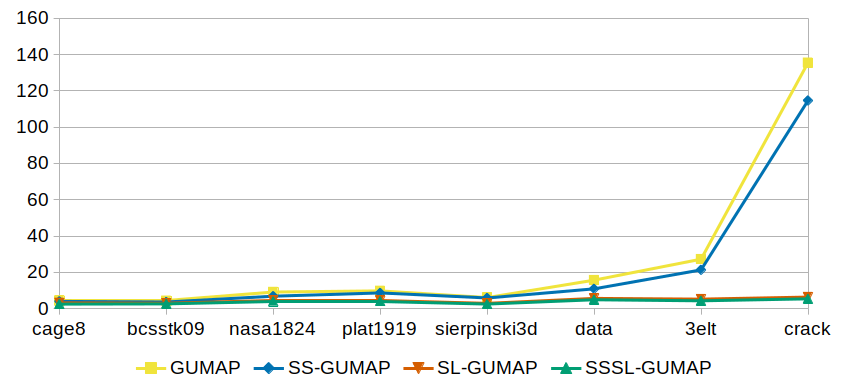}
        }
        \caption{Runtime comparison (in seconds): SS-GUMAP runs faster than GUMAP, and both SL-GUMAP and SSSL-GUMAP run significantly faster than GUMAP.}
    \label{fig:runtime_all}
    \vspace{-2mm}
\end{figure}

\subsection{SS-GUMAP Experiments}
\label{sec:umapss_exp}

In this experiment, we compare SS-GUMAP with GUMAP. 
For each graph $G$, we compute a spectral sparsification $G'$ with $O(n \log n)$ edges by selecting $O(n \log n)$ edges in decreasing order of effective resistance value. 
We then run UMAP on $G'$, and finally add back the edges removed in the sparsification to obtain a drawing of $G$. 
As spectral sparsification reduces the size of the graph, we expect SS-GUMAP to run faster than GUMAP. Furthermore, as spectral sparsification retains the most important edges, we expect SS-GUMAP to compute drawings with similar quality to GUMAP. We hypothesize the performance of SS-GUMAP:

\begin{hyp} \label{hyp:ss_runtime}
    SS-GUMAP runs significantly faster than GUMAP.
\end{hyp}

\begin{hyp} \label{hyp:ss_qual}
    SS-GUMAP computes similar quality drawings to GUMAP.
\end{hyp}

\subsubsection{Runtime}

Figure \ref{fig:runtime_all} shows the runtime of GUMAP in seconds in yellow, with the runtime of SS-GUMAP in blue.
On average, SS-GUMAP runs 28\% faster than GUMAP, supporting Hypothesis \ref{hyp:ss_runtime}. 
The largest runtime improvement is seen especially on larger, denser GION graphs such as GION\_4.

\subsubsection{Quality Metrics}
Figure \ref{fig:np_all} shows the neighborhood preservation, with GUMAP in yellow and SS-GUMAP in blue. 
On average, SS-GUMAP computes drawings with neighborhood preservation similar to drawings computed by GUMAP, at only 5\% lower, supporting Hypothesis \ref{hyp:ss_qual}. 
Note that the 28\% runtime gain is much higher than the 5\% quality loss.

In particular, for scale-free graphs, on average, SS-GUMAP obtains neighborhood preservation very similar to GUMAP, with only a 2\% difference. 
On GION and mesh graphs, the difference is still much lower, at only 7\% and 5\% respectively, compared to the 28\% runtime improvement.

\begin{figure}[t]
    \centering
    \subfloat[Scale-free]{
        \includegraphics[width=0.99\columnwidth]{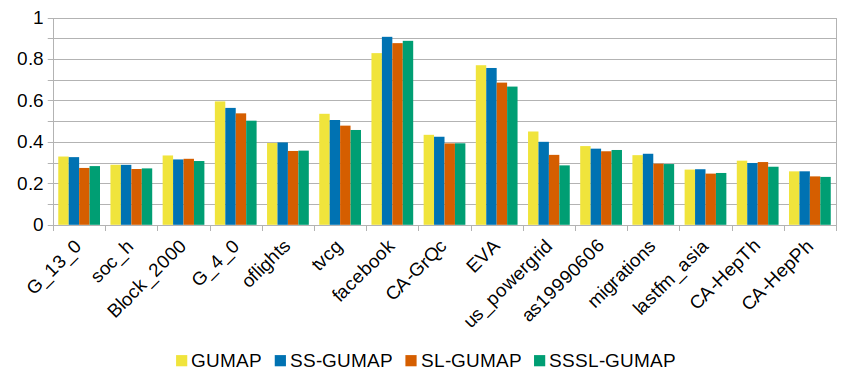}
        \label{fig:np_bench}
    }
    \qquad
    \subfloat[GION]{
        \includegraphics[width=0.78\columnwidth]{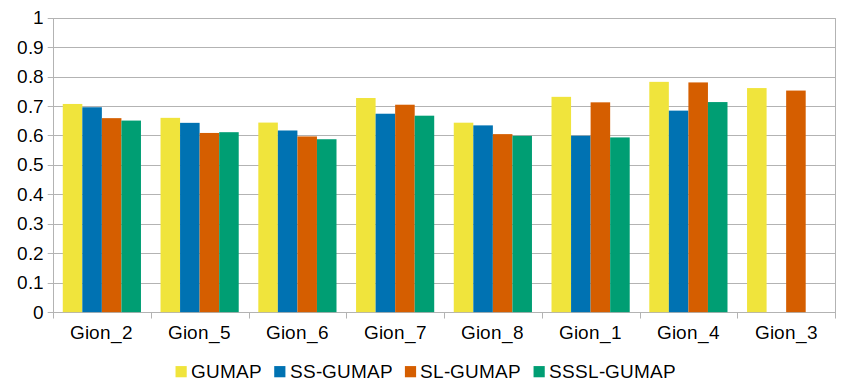}
        \label{fig:np_gion}
    }
    \qquad
    \subfloat[Mesh]{
        \includegraphics[width=0.78\columnwidth]{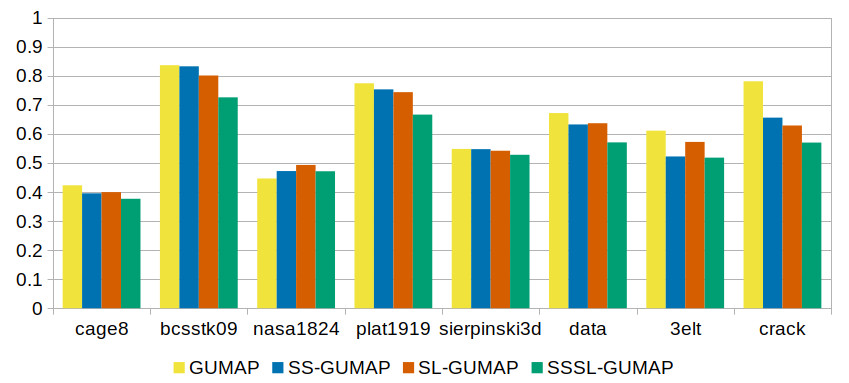}
        \label{fig:np_mesh}
    }
    \qquad
    \caption{Neighborhood preservation comparison: SS-GUMAP obtains very similar neighborhood preservation to GUMAP, while SL-GUMAP and SSSL-GUMAP obtain less than 11\% difference in neighborhood preservation than GUMAP.}
    \label{fig:np_all}
    \vspace{-2mm}
\end{figure}

\begin{figure}[t]
    \centering
    \subfloat[Scale-free]{
        \includegraphics[width=0.99\columnwidth]{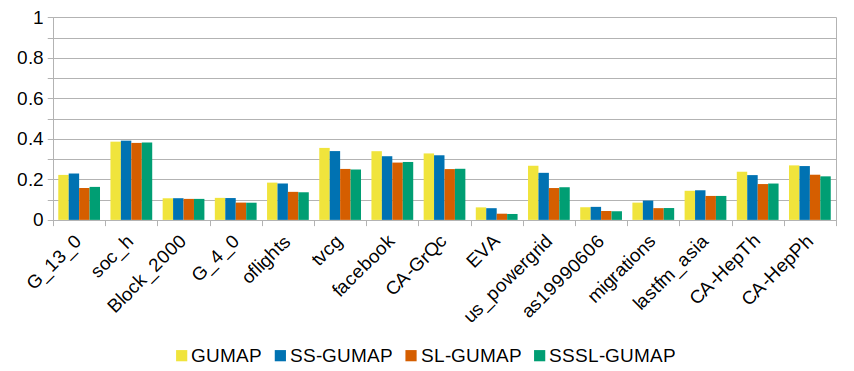}
        \label{fig:shp_bench}
    }
    \qquad
    \subfloat[GION]{
        \includegraphics[width=0.78\columnwidth]{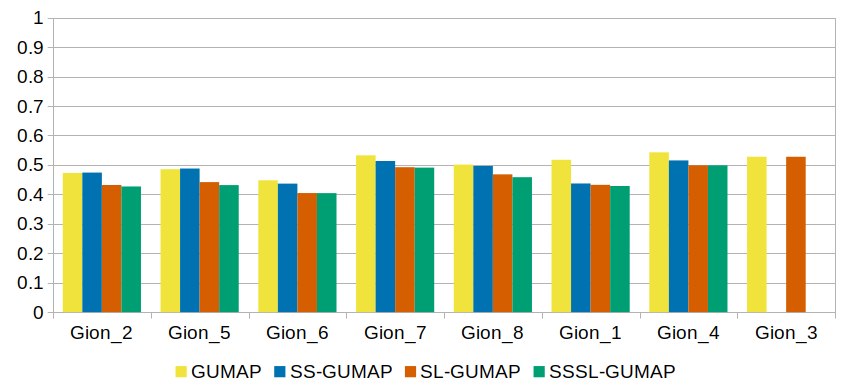}
        \label{fig:shp_gion}
    }
    \qquad
    \subfloat[Mesh]{
        \includegraphics[width=0.78\columnwidth]{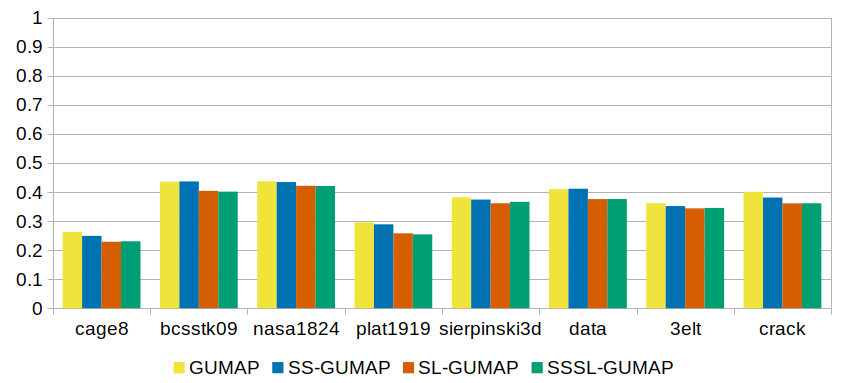}
        \label{fig:shp_mesh}
    }
    \qquad
    \caption{Shape-based metrics: on average, SS-GUMAP obtains the same shape-based metrics as GUMAP, while SL-GUMAP and SSSL-GUMAP obtain only 13\% lower shape-based metrics than GUMAP.}
    \label{fig:shp_all}
    \vspace{-2mm}
\end{figure}

Figure \ref{fig:shp_all} shows the shape-based metrics comparison, with GUMAP in yellow and SS-GUMAP in blue. 
On average, the shape-based metrics of drawings by SS-GUMAP is almost the same as GUMAP, supporting Hypothesis \ref{hyp:ss_qual}, at less than 4\% difference on average over all data sets, much lower than the average 28\% runtime improvement.

\begin{figure}[t]
    \centering
    \subfloat[Scale-free]{
        \includegraphics[width=0.99\columnwidth]{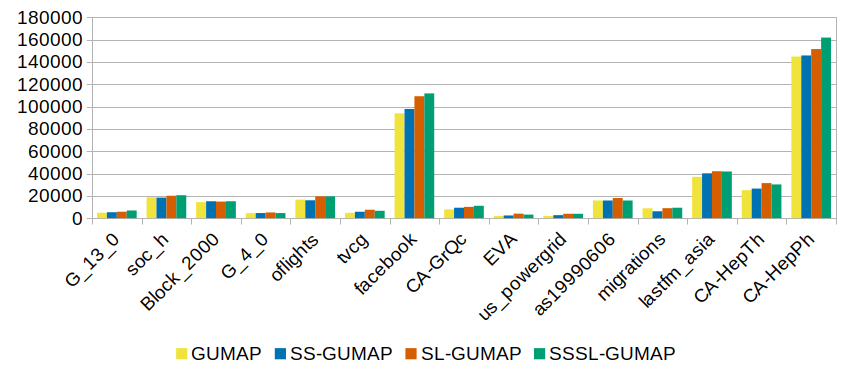}
        \label{fig:crossing_bench}
    }
    \qquad
    \subfloat[GION]{
        \includegraphics[width=0.78\columnwidth]{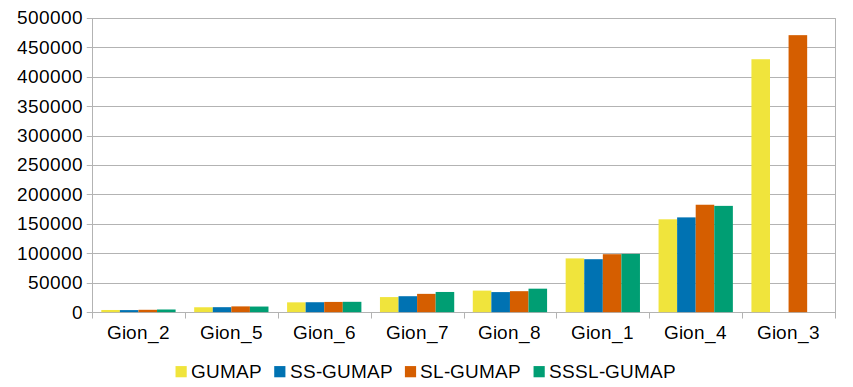}
        \label{fig:crossing_gion}
    }
    \qquad
    \subfloat[Mesh]{
        \includegraphics[width=0.78\columnwidth]{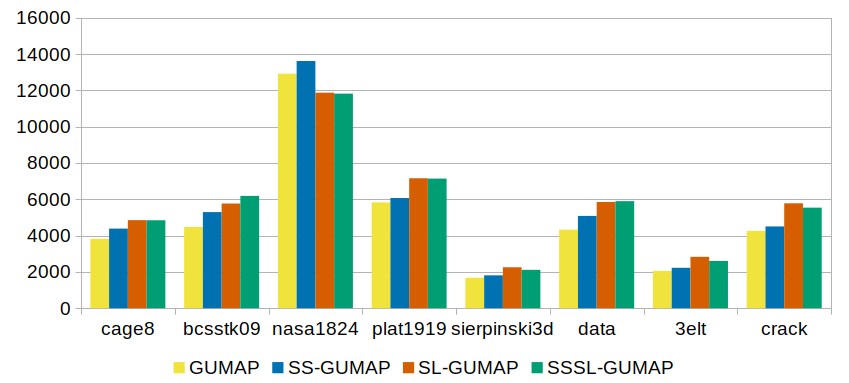}
        \label{fig:crossing_mesh}
    }
    \qquad
    \caption{Edge crossing: on average, SS-GUMAP obtains the same number of edge crossings as GUMAP, while SL-GUMAP and SSSL-GUMAP obtain around 12\% and 15\% more edge crossings than GUMAP, respectively.}
    \label{fig:crossing_all}
    \vspace{-2mm}
\end{figure}

Figure \ref{fig:crossing_all} shows the edge crossing of GUMAP in yellow and SS-GUMAP in blue. 
On average, SS-GUMAP computes drawings with about the same edge crossings as GUMAP, especially for scale-free and GION graphs, supporting Hypothesis \ref{hyp:ss_qual}.
For mesh graphs, on average, drawings by SS-GUMAP have only 9\% higher edge crossings than GUMAP, lower than the 17\% runtime improvement.

\begin{figure}[t]
    \centering
    \centering
    \subfloat[Scale-free]{
        \includegraphics[width=0.99\columnwidth]{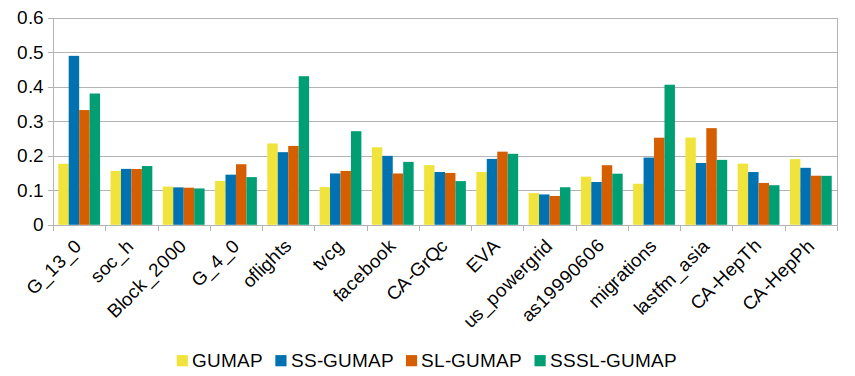}
        \label{fig:stress_bench}
    }
    \qquad
    \subfloat[GION]{
        \includegraphics[width=0.78\columnwidth]{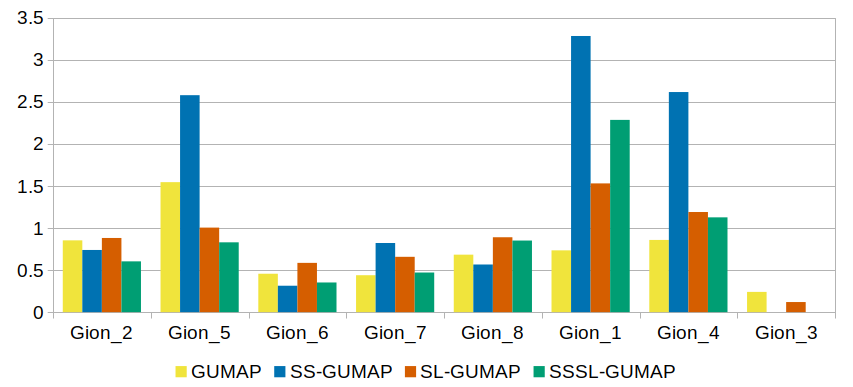}
        \label{fig:stress_gion}
    }
    \qquad
    \subfloat[Mesh]{
        \includegraphics[width=0.78\columnwidth]{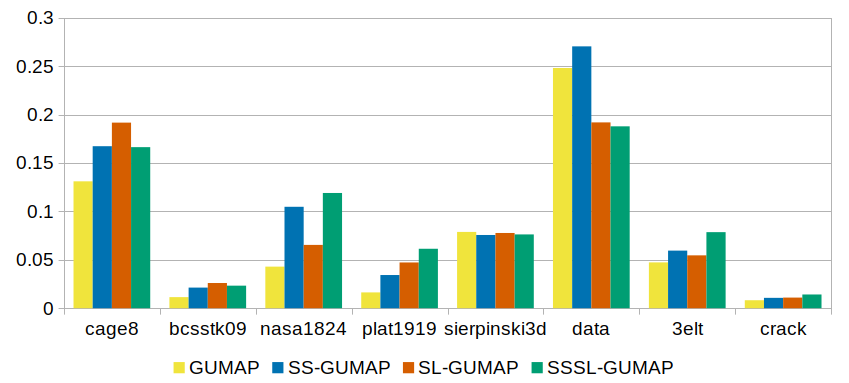}
        \label{fig:stress_mesh}
    }
    \qquad
    \caption{Stress comparison: SSSL-GUMAP obtains almost the same stress as GUMAP,
    excluding two outliers (GION 1 and GION 4), while SL-GUMAP obtains only 5\% higher stress on average. 
    SS-GUMAP obtains around 20\% higher stress than GUMAP on average, still much lower than the runtime improvement.}
    \label{fig:stress_all}
    \vspace{-2mm}
\end{figure}

Figure \ref{fig:stress_all} shows the stress comparison, with GUMAP in yellow and SS-GUMAP in blue, supporting Hypothesis \ref{hyp:sssl_qual} on average. 
For scale-free graphs, SS-GUMAP computes drawings with 11\% higher stress than GUMAP on average,  the same as the 11\% average runtime improvement on scale-free graphs, providing a comparable trade-off between runtime gain and quality loss.
The increase in stress is also mostly similar to the extent of runtime improvement on GION and mesh graphs, at both around 29\% and 18\% respectively, except for a few outliers, GION 1 and GION 4 for GION graphs, and nasa1824 for mesh graphs.

\subsubsection{Visual Comparison}

Table \ref{tab:viscomp} shows visual comparisons of the drawings by all four UMAP-based algorithms, on graphs of different types: scale-free graphs G\_13\_0 and Facebook, GION graphs GION\_6 and GION\_7, and mesh graphs data and plat1919. 
In general, SS-GUMAP produces drawings that are very similar to GUMAP on GION and mesh graphs. 

Furthermore, for some scale-free graphs, SS-GUMAP 
can better untangle the ``hairball'' compared to GUMAP; e.g., graph G\_13\_0 in the first row, where the dense region is expanded more in the drawing by SS-GUMAP, showing the inter-cluster structures more clearly.

\subsubsection{Discussion}

Experimental results support Hypothesis \ref{hyp:ss_runtime}, with SS-GUMAP running on average 28\% faster than GUMAP (see Figure \ref{fig:avg}(a)). 
On quality metrics, Hypothesis \ref{hyp:ss_qual} is supported especially for neighborhood preservation, shape-based metrics, and edge crossing. See Figures \ref{fig:avg}(b)-(d), where the averages on the three metrics are very similar between GUMAP and SS-GUMAP. 

Meanwhile, the higher stress of SS-GUMAP is mostly due to two outliers, GION\_1 and GION\_4: dense graphs with long diameters, forming a ``long'' structure without dense cluster structures. 
The sparsification process introduces artificial clusters, resulting in a drawing with overly-long edges between the artificially-introduced clusters. 

For scale-free graphs such as G\_13\_0 and Facebook in Table \ref{tab:viscomp}, SS-GUMAP is able to still display the clustering structures clearly; accordingly, the stress is only 11\% higher on average. 
Therefore, for stress, Hypothesis \ref{hyp:ss_qual} is most supported for scale-free and mesh graphs.

\subsection{SL-GUMAP Experiments}
\label{sec:exp_lumap}

In this experiment, we evaluate our SL-GUMAP algorithm against GUMAP for drawing a graph $G$. We use $k=15$ as the neighborhood size, and $n^{0.9}$ as the sample size for Step 2, the main gradient descent optimization loop. 
We expect SL-GUMAP to run faster than both GUMAP and SS-GUMAP, while maintaining similar quality to SS-GUMAP:

\begin{hyp} \label{hyp:sl_runtime}
    SL-GUMAP runs significantly faster than SS-GUMAP and GUMAP.
\end{hyp}

\begin{hyp} \label{hyp:sl_qual}
    SL-GUMAP computes similar quality drawings to GUMAP.
\end{hyp}

\subsubsection{Runtime}

On runtime, SL-GUMAP on average obtains 80\% improvement compared to GUMAP, also 70\% faster than SS-GUMAP, supporting Hypothesis \ref{hyp:sl_runtime}. 
As seen in Figure \ref{fig:runtime_all}, with SL-GUMAP shown in red and the graphs sorted in ascending order of vertex set size, the growth in runtime for SL-GUMAP as the size of the graphs grow is much slower than that of GUMAP, due to reducing the runtime of Steps 0-1 to $O(n)$ time and the runtime of Step 2 to sublinear-time.

\subsubsection{Quality Metrics}
On neighborhood preservation, as seen in Figure \ref{fig:np_all} with SL-GUMAP in red, drawings by SL-GUMAP obtain only 6\% lower neighborhood preservation than drawings by GUMAP, supporting Hypothesis \ref{hyp:sl_qual}.
Note that the reduction in neighborhood preservation is much lower than the average 80\% runtime improvement. 
In particular, on GION graphs, the neighborhood preservation of drawings by SL-GUMAP is only 4\% higher than drawings by GUMAP.

On shape-based metrics, SL-GUMAP also computes drawings with only 13\% difference to GUMAP, see Figure \ref{fig:shp_all}.
While the shape-based metrics are slightly lower than SS-GUMAP, the difference is still much smaller compared to the significant 80\% runtime improvement over GUMAP, supporting Hypothesis \ref{hyp:sl_qual}.

On edge crossings, SL-GUMAP computes drawings with only 13\% difference to GUMAP, see Figure \ref{fig:crossing_all}. 
With the relatively much smaller difference compared to the significant 80\% runtime improvement over GUMAP, the results support Hypothesis \ref{hyp:sl_qual}.

On stress, not only does SL-GUMAP compute drawings with stress that is close to GUMAP, it also computes drawings with lower stress than SS-GUMAP, supporting Hypothesis \ref{hyp:sl_qual}. See Figure \ref{fig:stress_all}.
On average, drawings by SL-GUMAP obtain only 16\% higher stress than by GUMAP, still much lower than the 80\% runtime improvement, while also obtaining better stress than SS-GUMAP. 
Moreover, when excluding two outliers, GION 1 and GION 4, the average stress is only 5\% higher than GUMAP.

\subsubsection{Visual Comparison}

As shown in Table \ref{tab:viscomp}, with drawings by SL-GUMAP in the third column, SL-GUMAP can compute drawings with comparable quality to GUMAP for scale-free and GION graphs. 
With scale-free graphs, SL-GUMAP is also able to compute drawings which display the intra-cluster structure better than GUMAP, as seen in the Facebook graph example: the clusters are less overly condensed than GUMAP, showing better intra-cluster structure without excessive distortion of the global inter-cluster structure. 
It can also compute drawings with fewer overly long edges, as with GION\_6 and GION\_7.

\begin{table*}[ht!]
    \centering
    \caption{Visual comparison: In general, SS-GUMAP produces drawings with similar quality to UMAP. Meanwhile, SL-GUMAP and SSSL-GUMAP compute drawings with fewer overly-long edges with better intra-cluster structure, better untangling than GUMAP on scale-free and GION graphs. }
    \begin{tabular}{|c|c|c|c|}
    \hline
    GUMAP & SS-GUMAP & SL-GUMAP & SSSL-GUMAP \\ \hline
    \multicolumn{4}{|c|}{G\_13\_0} \\ \hline
    \includegraphics[width=0.36\columnwidth]{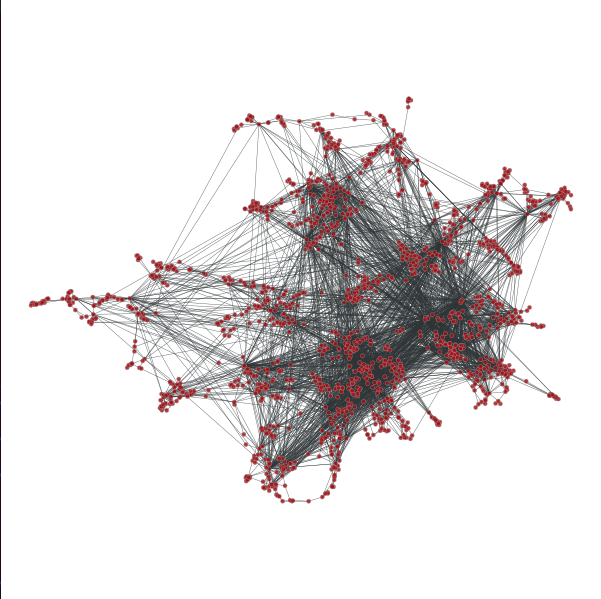} &
    \includegraphics[width=0.36\columnwidth]{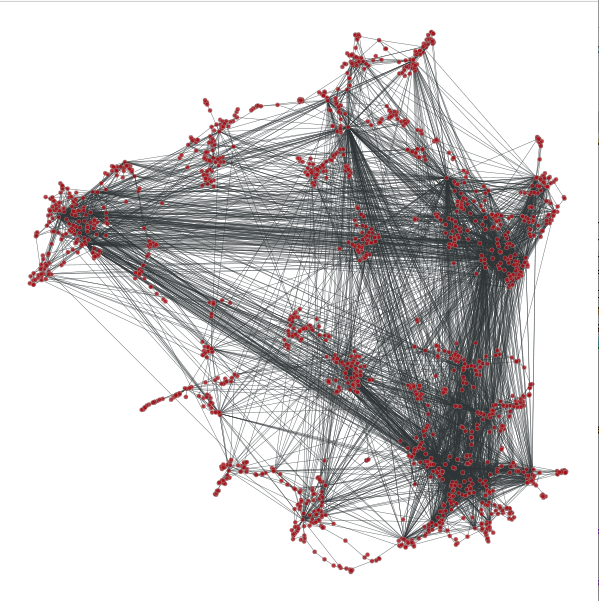} &
    \includegraphics[width=0.36\columnwidth]{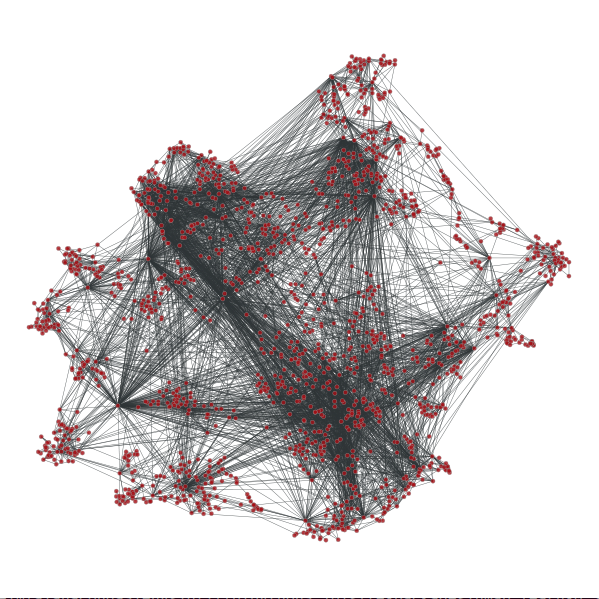} &
    \includegraphics[width=0.36\columnwidth]{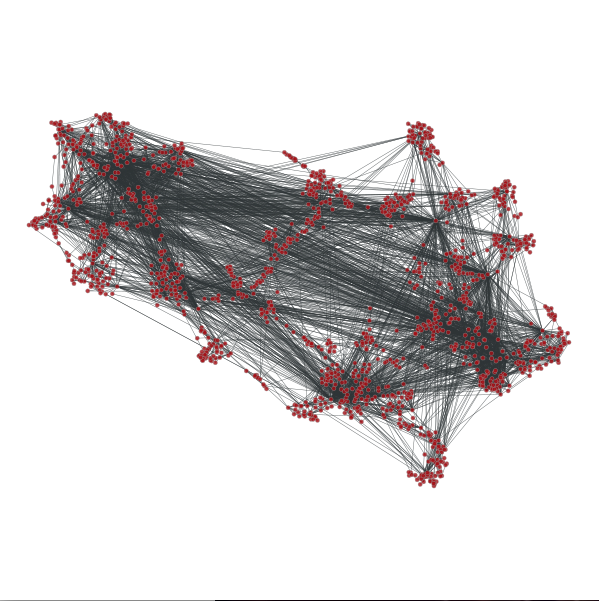} \\ \hline
    \multicolumn{4}{|c|}{Facebook} \\ \hline
    \includegraphics[width=0.36\columnwidth]{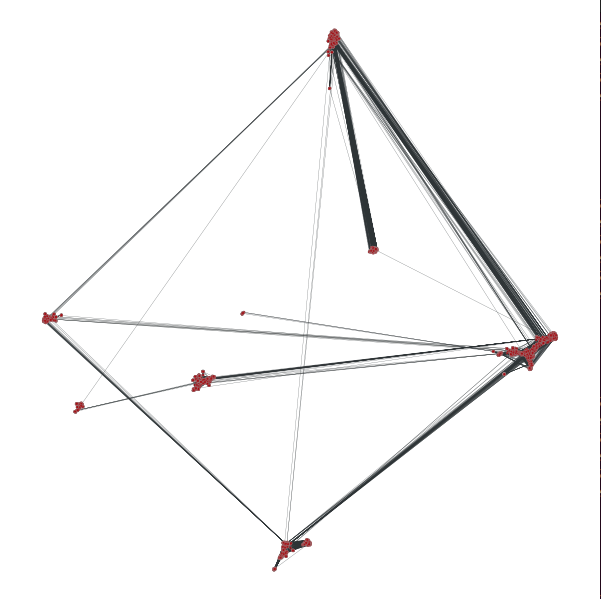} &
    \includegraphics[width=0.36\columnwidth]{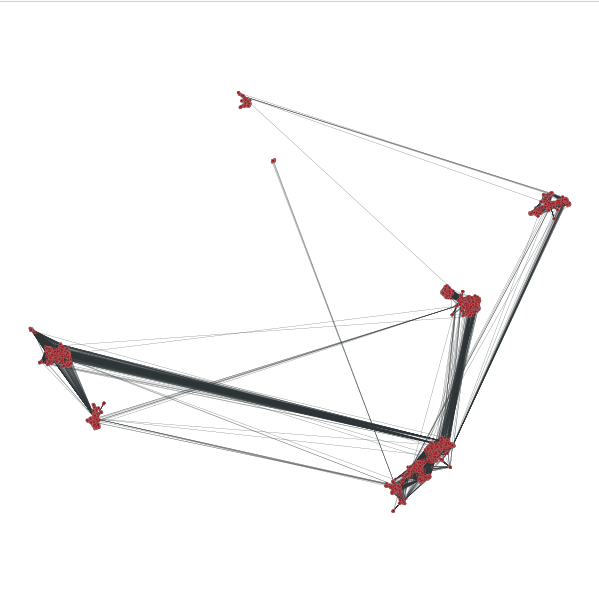} &
    \includegraphics[width=0.36\columnwidth]{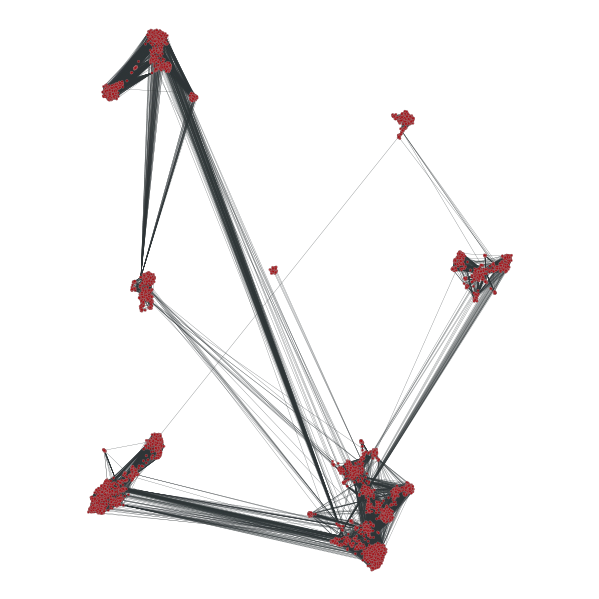} &
    \includegraphics[width=0.36\columnwidth]{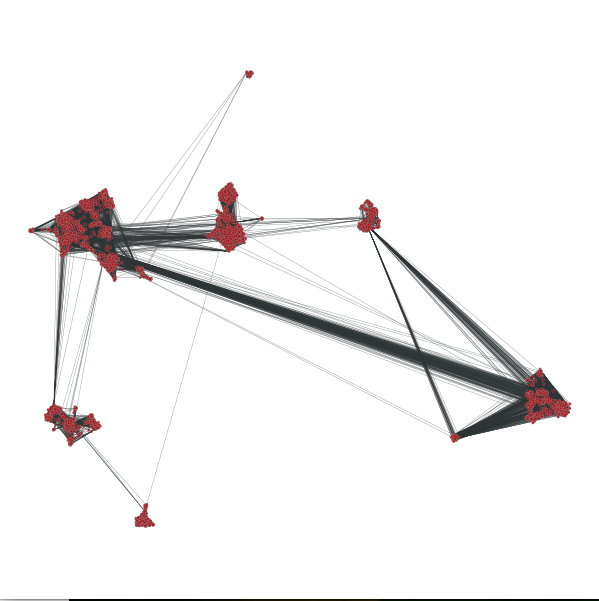} \\ \hline
    \multicolumn{4}{|c|}{GION\_6} \\ \hline
    \includegraphics[width=0.36\columnwidth]{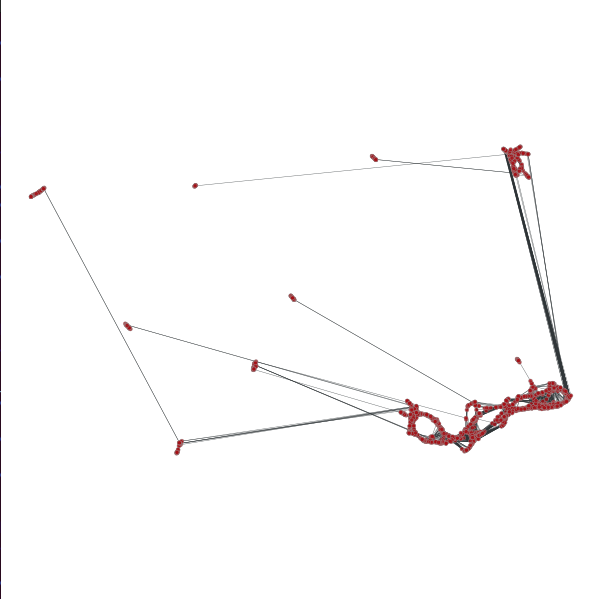} &
    \includegraphics[width=0.36\columnwidth]{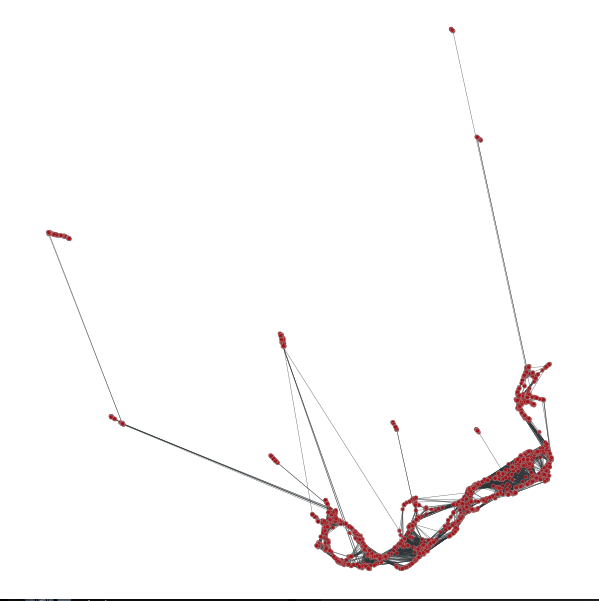} &
    \includegraphics[width=0.36\columnwidth]{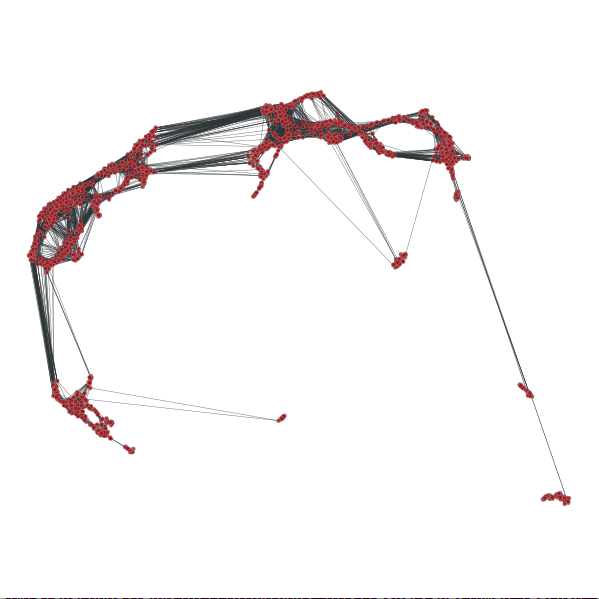} &
    \includegraphics[width=0.36\columnwidth]{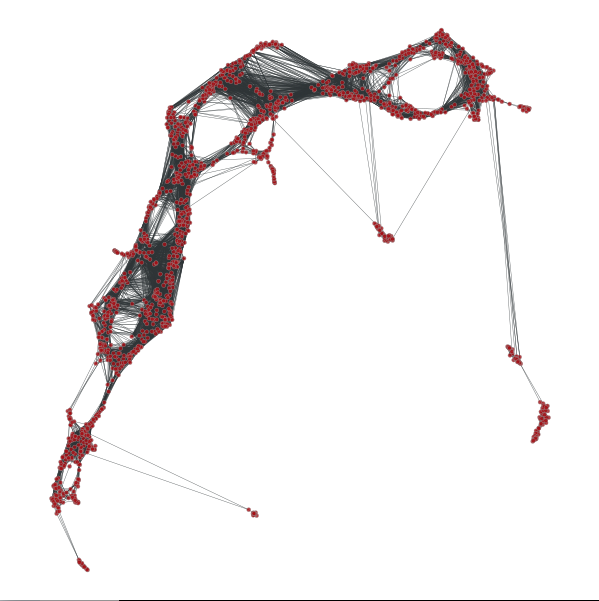} \\ \hline
    \multicolumn{4}{|c|}{GION\_7} \\ \hline
    \includegraphics[width=0.36\columnwidth]{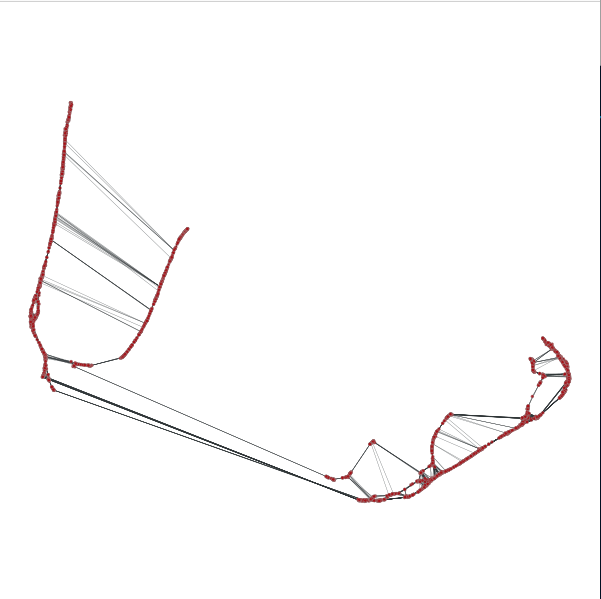} &
    \includegraphics[width=0.36\columnwidth]{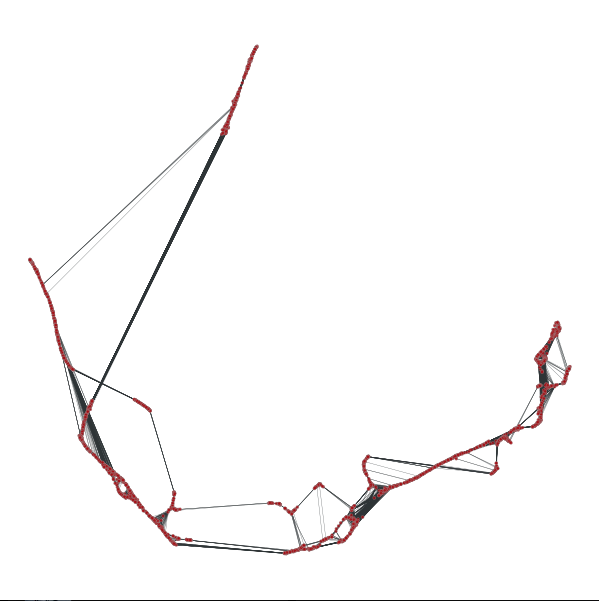} &
    \includegraphics[width=0.36\columnwidth]{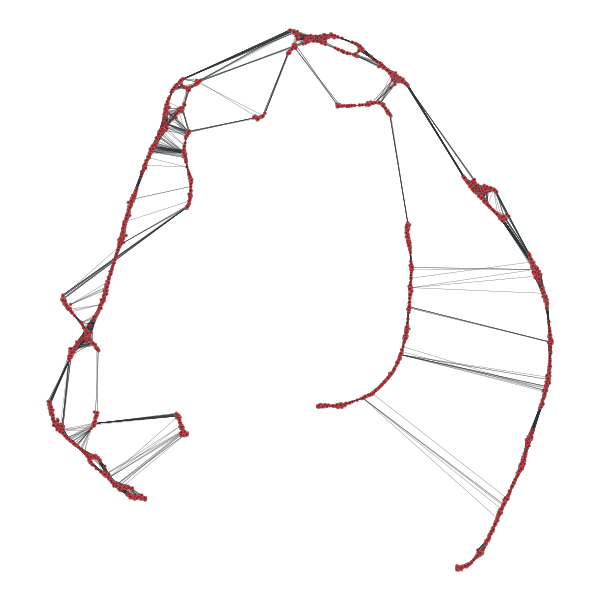}&
    \includegraphics[width=0.36\columnwidth]{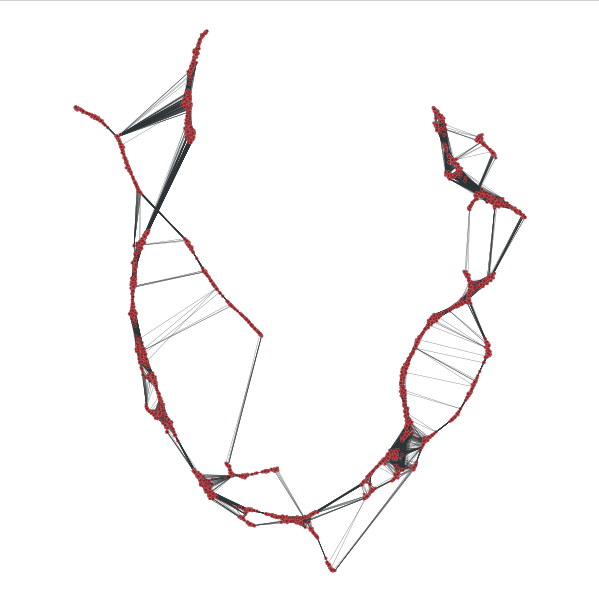} \\ \hline
    \multicolumn4{|c|}{data} \\ \hline
    \includegraphics[width=0.36\columnwidth]{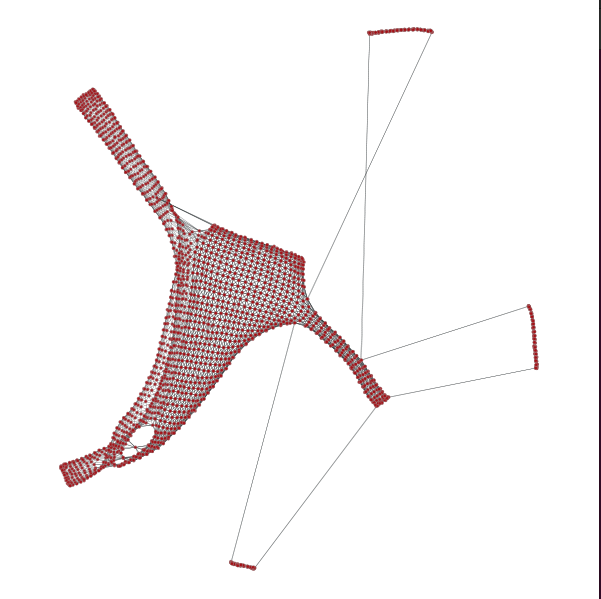} &
    \includegraphics[width=0.36\columnwidth]{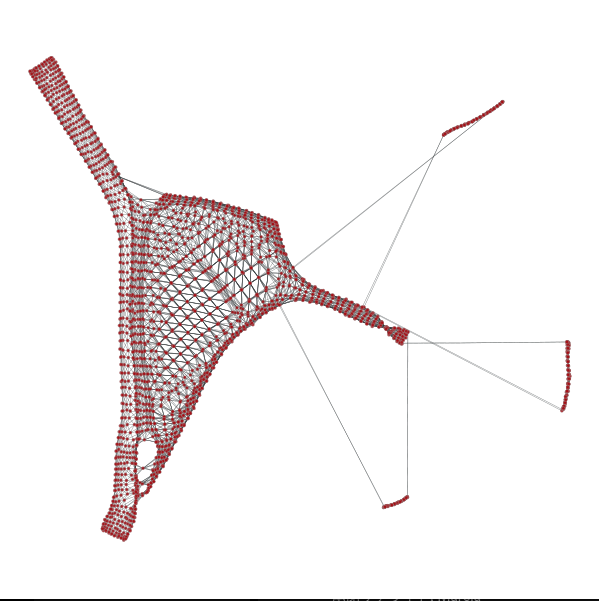} &
    \includegraphics[width=0.36\columnwidth]{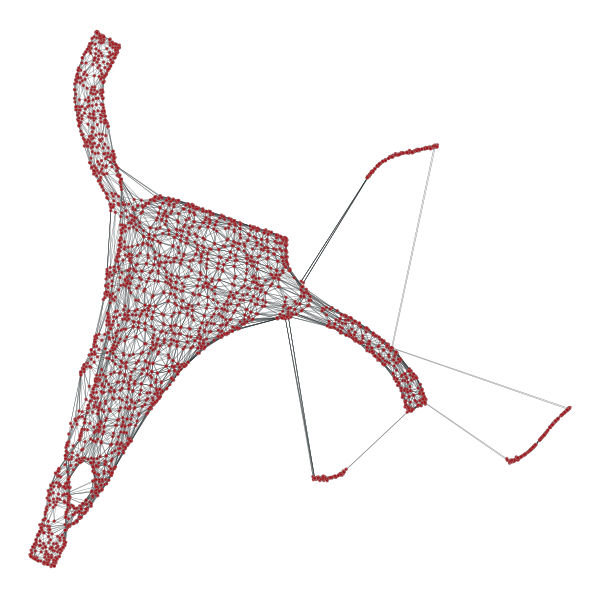} &
    \includegraphics[width=0.36\columnwidth]{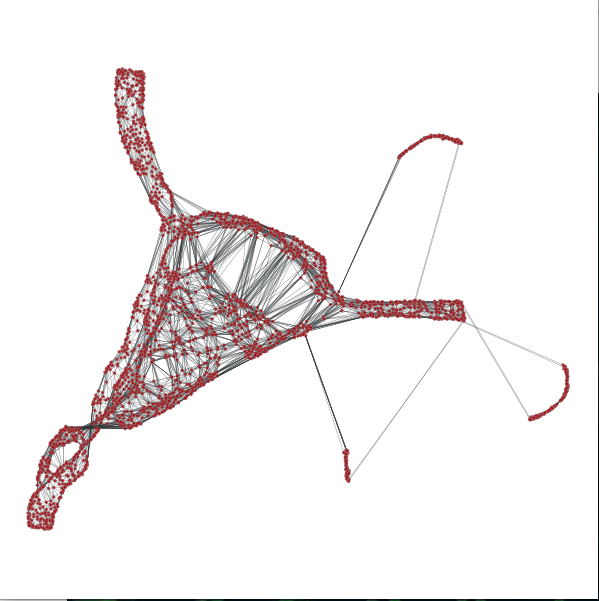} \\ \hline
    \multicolumn4{|c|}{plat1919} \\ \hline
    \includegraphics[width=0.36\columnwidth]{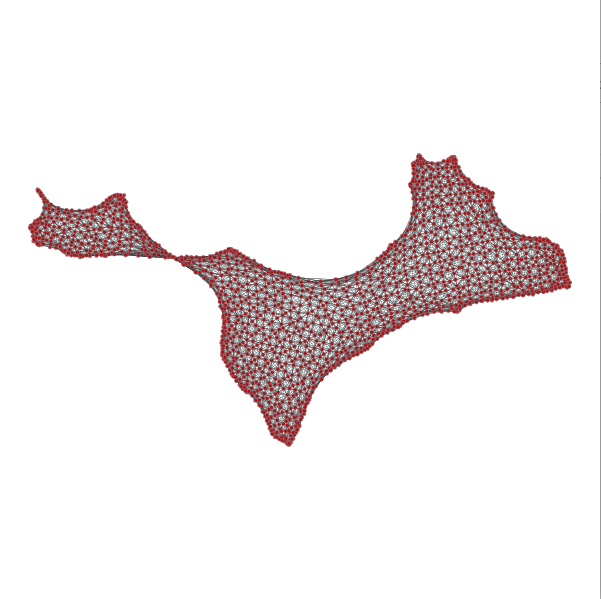} &
    \includegraphics[width=0.36\columnwidth]{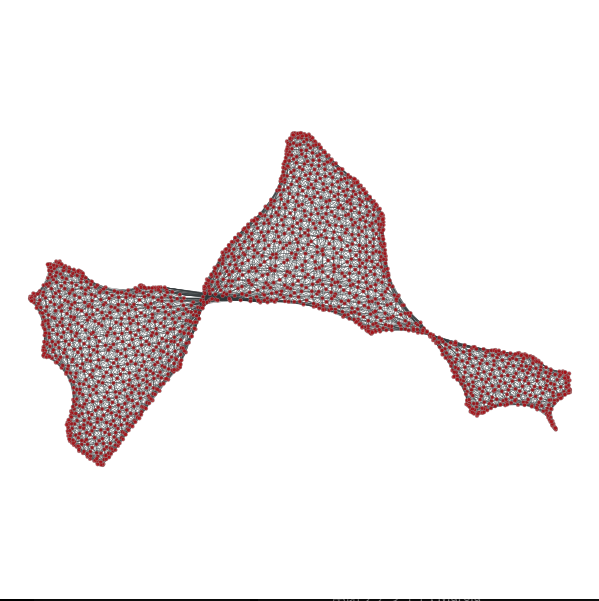} &
    \includegraphics[width=0.36\columnwidth]{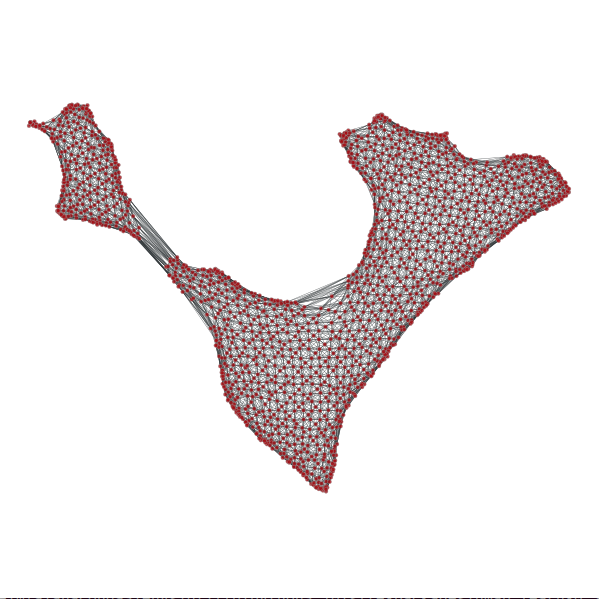} &
    \includegraphics[width=0.36\columnwidth]{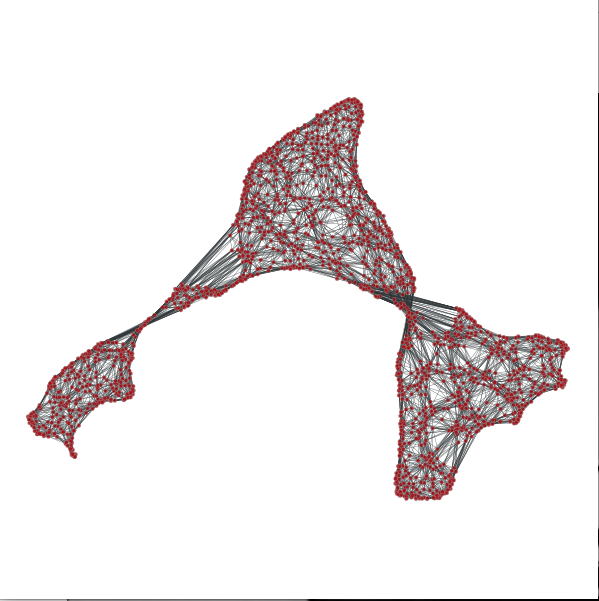} \\ \hline
    \end{tabular}
    \label{tab:viscomp}
\end{table*}

\subsubsection{Discussion}

Experiments support Hypothesis \ref{hyp:sl_runtime}, showing that SL-GUMAP runs significantly faster than GUMAP at 80\% faster on average, see Figure \ref{fig:avg}(a). 
On quality metrics, Hypothesis \ref{hyp:sl_qual} is more supported on neighborhood preservation and stress, with less than 7\% difference to the metrics of drawings by GUMAP, see Figures \ref{fig:avg}(b) and (e), 
much lower than the runtime improvement. 
Even on shape-based metrics and edge crossing, the 13\% difference (see Figure \ref{fig:avg}(c)-(d)) is still much lower than the 80\% runtime improvement, providing a favorable runtime-quality trade-off. 

Surprisingly, SL-GUMAP obtains lower stress on average than SS-GUMAP, especially on GION and mesh graphs. 
From the visual comparison in Table \ref{tab:viscomp}, on some mesh graphs, SS-GUMAP would ``twist'' the global mesh shape that is preserved by SL-GUMAP, while on GION graphs, drawings by SL-GUMAP contain fewer overly long edges compared to SS-GUMAP; these may have led to lower stress for drawings by SL-GUMAP.

\subsection{SSSL-GUMAP Experiments}

In this experiment, we compare UMAP with SSSL-GUMAP. We first compute the spectral sparsification $G'$ of a graph $G$, and then run SL-GUMAP on $G'$ with the same settings as with the experiments in Section \ref{sec:exp_lumap}, $k=15$ as the neighborhood size, and $n^{0.9}$ as the sample size for Step 2. 
We hypothesize the performance of SL-GUMAP:

\begin{hyp} \label{hyp:sssl_runtime}
    SSSL-GUMAP runs faster than SL-GUMAP.
\end{hyp}
\begin{hyp} \label{hyp:sssl_qual}
    SSSL-GUMAP computes similar quality drawings to GUMAP.
\end{hyp}

\subsubsection{Runtime}


SSSL-GUMAP runs the fastest of all our algorithms, on average 83\% faster than GUMAP,  78\% faster than SS-GUMAP, and 23\%  faster than SL-GUMAP, supporting Hypothesis \ref{hyp:sssl_runtime}. 
See Figure \ref{fig:runtime_all}, with SSSL-GUMAP in green.

The largest runtime improvement is seen on the mesh graphs, with 86\% improvement over UMAP on average. 
The growth in runtime as the size of the graph increases is much slower compared to UMAP and SS-GUMAP, see the slope of the green curve (SSSL-GUMAP) compared to the yellow curve (GUMAP) in Figure \ref{fig:runtime_all}. 

\subsubsection{Quality Metrics}
On neighborhood preservation, SSSL-GUMAP obtains, on average, only 11\% higher neighborhood preservation, as seen in Figure \ref{fig:np_all}, supporting Hypothesis \ref{hyp:sssl_qual}. 
Especially on scale-free and GION graphs, the neighborhood preservation is less than 10\% lower than that of GUMAP, much smaller than the 83\% runtime improvement.

On shape-based metrics, SSSL-GUMAP obtains on average a 13\% difference in metrics compared to GUMAP (Figure \ref{fig:shp_all}), supporting Hypothesis \ref{hyp:sssl_qual}.
SSSL-GUMAP obtains less than 8\% lower shape-based metrics on mesh graphs, and 10\% lower on GION graphs.

For edge crossing metrics, SSSL-GUMAP obtains on average a 15\% difference compared to GUMAP, see Figure \ref{fig:crossing_all}. 
While slightly higher than SL-GUMAP, the difference is much lower than the 83\% runtime improvement over GUMAP, supporting Hypothesis \ref{hyp:sssl_qual}.

On stress, SSSL-GUMAP obtains on average just 20\% higher stress than GUMAP, much lower than the 80\% runtime improvement; 
see Figure \ref{fig:stress_all}.
Notably, on GION graphs, except  
GION 1 and GION 4, which also show outlier behaviour on both SS-GUMAP and SL-GUMAP, SSSL-GUMAP obtains even lower stress than GUMAP, supporting Hypothesis \ref{hyp:sssl_qual}. 
On mesh graphs, except for nasa1824, 
SSSL-GUMAP, on average, obtains only 12\% higher stress than GUMAP, much lower than the 83\% runtime improvement.

\subsubsection{Visual Comparison}

On visual comparison, SSSL-GUMAP tends to display similar strengths and shortcomings as SL-GUMAP, as seen in the last column of Table \ref{tab:viscomp}. 
On GION graphs, SSSL-GUMAP further reduces the number of overly long edges compared to SL-GUMAP. SSSL-GUMAP also shows better untangling of intra-cluster structures compared to GUMAP on some scale-free graphs without excessive distortion of the global/inter-cluster structure, such as the Facebook graph, where the clusters are not as ``compressed'' or the GION\_6 graphs, where the cycles are more visible without folding up the long diameter structure. However, on mesh graphs, it sometimes maintains the global structure, but not the internal regular grid, as seen in the data graph.

\subsubsection{Discussion}

Experiments support Hypothesis \ref{hyp:sssl_runtime}, showing that SSSL-GUMAP runs 83\% faster than GUMAP on average (see Figure \ref{fig:avg}(a)). 
The quality metrics of drawings by SSSL-GUMAP are slightly lower than SL-GUMAP, although the loss is still much smaller than the significant runtime improvement, overall supporting Hypothesis \ref{hyp:sssl_qual}.

In particular, excluding two outliers GION\_1 and GION\_4, stress is almost on the same level as GUMAP and lower than SS-GUMAP (see Figure \ref{fig:avg}(e)). 
On neighbourhood preservation, shape-based metrics, and edge crossings, the difference is also much smaller than the 83\% runtime improvement, at only 13\%, 11\%, and 15\% respectively (see Figure \ref{fig:avg}(b)-(d)).


\subsection{Summary and Recommendation}

Our experiments demonstrate the strengths of our fast UMAP algorithms for graph drawing, running significantly faster than the base GUMAP with a favorable trade-off between runtime and quality.

Our algorithms especially perform well on neighborhood preservation, where the drawings computed by our algorithms obtain similar neighborhood preservation to that of GUMAP.
As one of the strengths of UMAP is the preservation of point neighborhoods, this supports that our algorithms are able to preserve the main strength of UMAP with much faster runtime.

On visual comparison, our algorithms tend to work better on graphs with globally sparse and locally dense structures, such as the scale-free graphs and most GION graphs. 
A strength often seen over GUMAP is better untangling of intra-cluster structures for scale-free graphs, balanced with preservation of the global structure, especially visible for SL-GUMAP and SSSL-GUMAP (see the Facebook graph in Table \ref{tab:viscomp}). 
This may be due to the sampling reducing the number of pairs of neighbors on which $w_{ij}$ is computed each iteration, creating a ``repulsion'' effect between neighbors.

The runtime and quality metric comparisons also show different strengths between our algorithms. 
On runtime, SSSL-GUMAP is the clear winner, always obtaining the highest runtime improvement. 
On quality metrics, the best-performing algorithms differ per data set. 
For scale-free graphs, SS-GUMAP computes drawings with the most similar quality metrics to GUMAP overall. 
For GION and mesh graphs, SL-GUMAP computes drawings with the most similar neighborhood preservation and stress to GUMAP, while SS-GUMAP performs better on edge crossing and shape-based metrics.

The better performance of SL-GUMAP compared to SS-GUMAP on neighborhood preservation and stress for GION and mesh graphs may be due to SS-GUMAP's tendency to artificially introduce clusters in data sets without a dense clustering structure. 
An example can be seen in the outliers of the GION data set, GION 1 and GION 4, where even the drawing by GUMAP contains some overly long edges connecting artificially introduced clusters. 
This tendency is much more amplified in the drawing by SS-GUMAP compared to SL-GUMAP; the removal of edges may have created subgraphs having a higher density than the surroundings, artificially creating clusters.



\begin{figure}[h!]
    \centering
    \subfloat[Runtime]{
        \includegraphics[width=0.3\columnwidth]{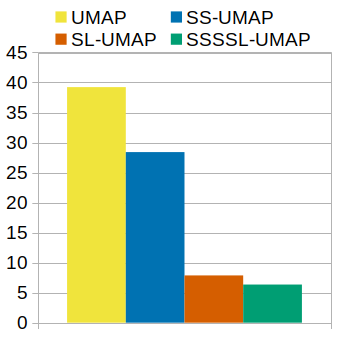}
        \label{fig:runtime_avg}
    }
    \subfloat[Neigh. preserv.]{
        \includegraphics[width=0.3\columnwidth]{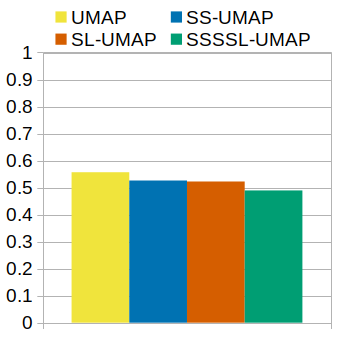}
        \label{fig:rnp_avg}
    }
    \subfloat[Shape-based]{
        \includegraphics[width=0.3\columnwidth]{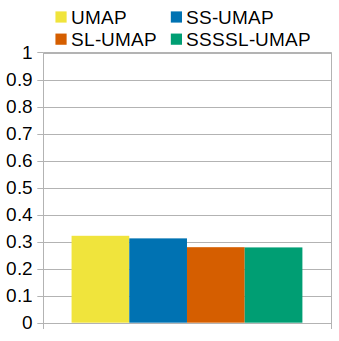}
        \label{fig:shp_avg}
    }
    \qquad
    \subfloat[Crossing]{
        \includegraphics[width=0.3\columnwidth]{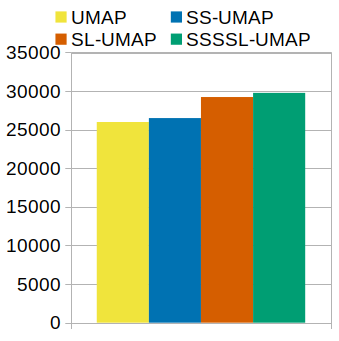}
        \label{fig:crossing_avg}
    }
        \subfloat[Stress]{
        \includegraphics[width=0.3\columnwidth]{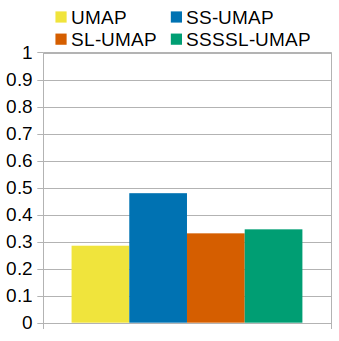}
        \label{fig:stress_avg}
    }
    \caption{Average runtime and quality metrics: our algorithms run significantly faster than GUMAP with much smaller reduction in quality metrics, obtaining much higher runtime gain over quality loss.}
    \label{fig:avg}
    \vspace{-2mm}
\end{figure}

\emph{In summary, we validate our hypotheses that SS-GUMAP, SL-GUMAP, and SSSL-GUMAP run significantly faster than GUMAP, on average 28\%, 80\%, and 83\% faster, respectively (see Figure \ref{fig:avg}(a)).
Furthermore, we show that for all quality metrics, any reduction is far less than the significant runtime improvement, providing a good runtime-quality trade-off (see the averages in Figure \ref{fig:avg}(b)-(e))}.

We recommend the following guidelines: if runtime is of the most concern, SSSL-GUMAP runs the fastest, while for better quality with good runtime, SS-GUMAP is recommended for scale-free graphs, while SL-GUMAP is recommended for mesh-like graphs.

\section{GUMAP vs tsNET}
\label{sec:ts_umap}

In order to place the study of UMAP for graph drawing in context, we now present our experiments comparing GUMAP to \emph{tsNET}~\cite{kruiger2017graph}, a graph drawing algorithm based on t-SNE~\cite{van2008visualizing}, a popular DR algorithm focusing on neighborhood preservation.


We focus on the comparison between tsNET and GUMAP: both UMAP and t-SNE focus on \emph{neighborhood preservation}, and tsNET, the baseline method for t-SNE-based graph drawing, has been compared to other popular graph drawing layouts, both DR-based or otherwise, including Pivot MDS, FM\textsuperscript{3}~\cite{hachul2004drawing}, sfdp~\cite{hu2005efficient}, and Stress Majorization~\cite{gansner2005graph}, and has been shown to obtain superior neighborhood preservation. While newer t-SNE algorithms have been presented, we focus on tsNET, which has been designed and implemented specifically for graph drawing. Furthermore, we focus on the comparison of tsNET as it is designed for drawing general graphs, as opposed to methods such as GraphTSNE~\cite{leow2019graphtsne}, which is designed for graphs with vertex properties.

Although a recent study provided a comparison of UMAP with tsNET~\cite{miller2023balancing}, tsNET with \emph{random} initialization was used. Pivot MDS (PMDS) initialization for tsNET has been shown to be both fast and to produce better quality results for tsNET~\cite{kruiger2017graph}. Thus, a fairer comparison with the respective suitable initializations for both tsNET and GUMAP would be important to compare the full performance of both algorithms. We also compare implementations in the same language (Python) to minimize discrepancies in efficiency due to technical constraints.

In addition, we focus on a baseline-to-baseline comparison between GUMAP and tsNET. Although DRGraph, another t-SNE-based graph drawing algorithm, has been presented~\cite{zhu2020drgraph}, it has been compared with tsNET, with similar quality to tsNET on average~\cite{zhu2020drgraph,ltsnet}. Thus, it is sufficient to compare the quality metrics of GUMAP to tsNET.

Note that we have not used GPU-based algorithms for our algorithms. Our aim is for a fast algorithm with low theoretical runtime, and many GPU-based implementations may obscure whether the runtime gain is obtained through parallelization or from reducing the theoretical runtime complexity.

\subsection{tsNET}

The \emph{tsNET}~\cite{kruiger2017graph} algorithm utilizes t-SNE for graph drawing. The original t-SNE algorithm computes a low-dimensional projection preserving the neighborhoods of data points by modelling pairwise similarities in the high- and low-dimensional space as probability distributions, then minimizing the Kullback-Leibler (KL) divergence between the two probability distributions.

Unlike t-SNE, which only optimizes one term for its cost function, tsNET optimizes a cost function $C$ consisting of the following three terms:

\begin{equation}
\label{eq:tsnet}
\begin{split}
C = C_{KL} + C_{CMP} + C_{ENT}, \\
C_{KL} = \lambda_{KL} \sum_{i \neq j} p_{ij} \log \frac{p_{ij}}{q_{ij}},\\
C_{CMP} = \frac{\lambda_c}{2n}\sum_i ||X_i||^2, \\
C_{ENT} = - \frac{\lambda_r}{2n^2} \sum_{i \neq j} \log (||X_i - X_j|| + \epsilon)
\end{split}
\end{equation}

The first term $C_{KL}$ is the same KL divergence term used in t-SNE, computed using the high-dimensional probability $p_{ij}$ and low-dimensional probability $q_{ij}$. 
The second term $C_{CMP}$ is \emph{compression}~\cite{van2008visualizing}, added to speed up convergence. The third term $C_{ENT}$ is \emph{entropy}~\cite{gansner2012maxent}, added to avoid clutter. The cost function of tsNET can be optimized in $O(n^2)$ time.

tsNET has been evaluated against well-known layout algorithms, including SFDP~\cite{hu2005efficient}, LinLog~\cite{noack2003energy}, GRIP~\cite{gajer2002grip}, and NEATO~\cite{north2004drawing}, outperforming the other layouts on neighborhood preservation~\cite{kruiger2017graph}.

\subsection{Experiment Design}


We use the tsNET implementation based on the implementations from Kruiger et al.~\cite {kruiger2017graph} and scikit-learn~\cite{pedregosa2011scikit}. For tsNET, we use Pivot MDS (PMDS)~\cite{brandes2006eigensolver} initialization, per the recommendation of the tsNET authors~\cite{kruiger2017graph}.

For each graph, perplexity is set as the minimum between 40 and the smallest multiple of 100 that does not produce a ``perplexity too low'' error. Both tsNET and UMAP are run either until convergence is reached, or the maximum number of iterations, set at 500 based on the maximum of the default values for tsNET and UMAP, is reached.


We use the same data sets as used in Section \ref{sec:eval} for our experiments, separated into three sets: benchmark \emph{real-world scale-free} graphs, with a globally sparse and locally dense structure~\cite{snapnets}; biochemical \emph{GION} graphs with large diameters~\cite{marner2014gion}; and (3) \emph{mesh} graphs with regular grid-like structures~\cite{davis2011university}. See Table \ref{table:data} for details.


To evaluate the algorithms, we select popular quality metrics, taken from both graph drawing and dimension reduction studies, the same as used in Section \ref{sec:eval}: neighborhood preservation, edge crossing, shape-based metrics, and stress. 
%
For each experiment, we take the average runtime and metrics over five runs.

\begin{table*}[h!]
    \centering
    \caption{Visual comparison of GUMAP and tsNET: in general, GUMAP tends to group clusters closer together, while tsNET tends to ``expand'' the local neighborhood of vertices more.}
    \begin{tabular}{|c|c|c|c|}
    \hline
    GUMAP & tsNET & GUMAP & tsNET \\ \hline
    \multicolumn{2}{|c|}{as19990606} & \multicolumn{2}{|c|}{G\_4\_0} \\ \hline
    \includegraphics[width=0.36\columnwidth]{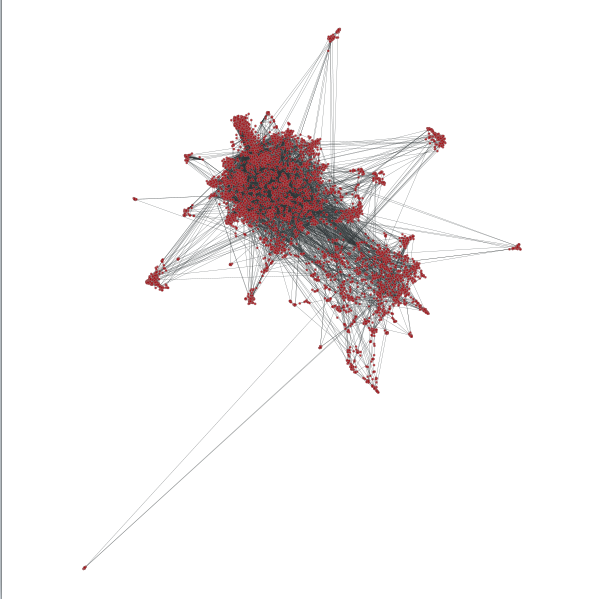} &
    \includegraphics[width=0.36\columnwidth]{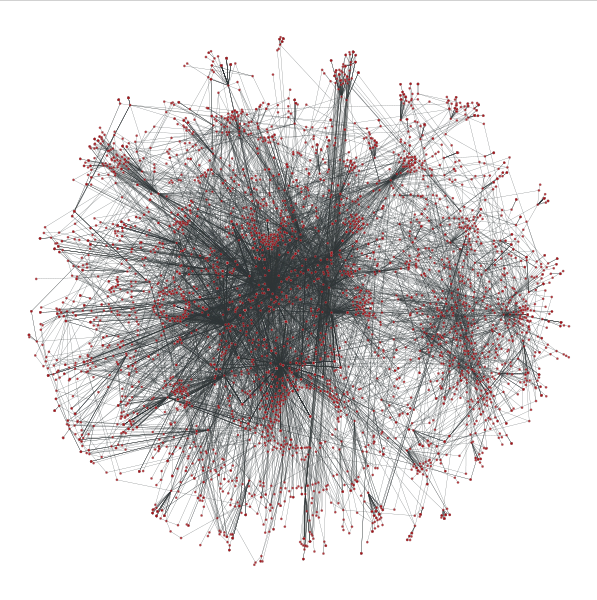} &
    \includegraphics[width=0.36\columnwidth]{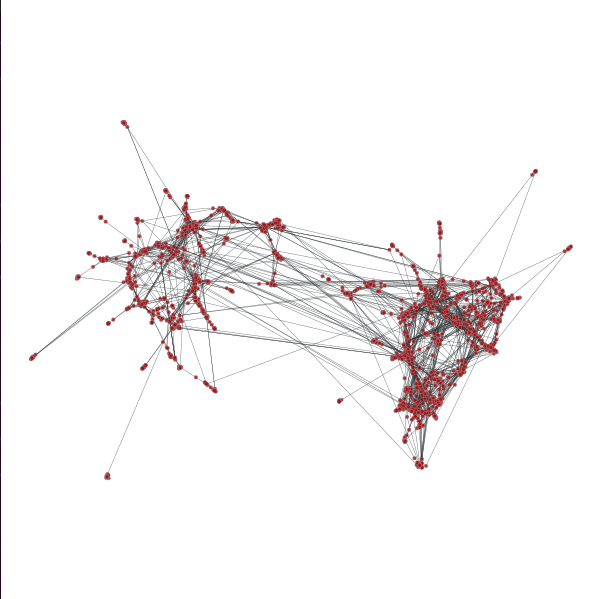} &
    \includegraphics[width=0.36\columnwidth]{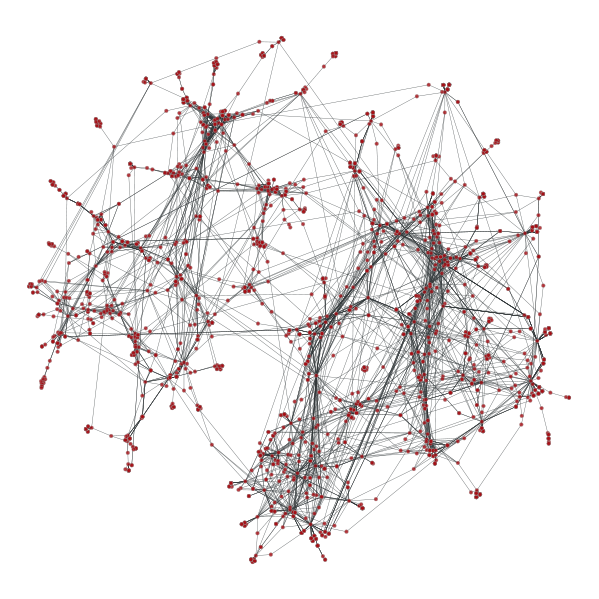} \\ \hline
    \multicolumn{2}{|c|}{GION\_1} & \multicolumn{2}{|c|}{plat1919} \\ \hline
    \includegraphics[width=0.36\columnwidth]{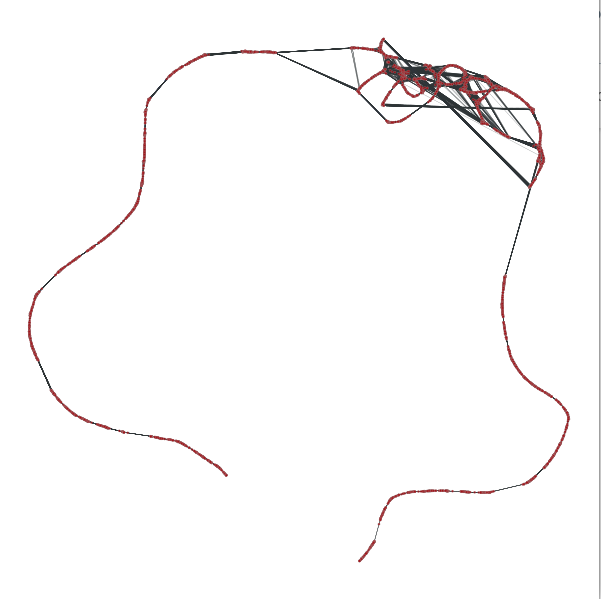} &
    \includegraphics[width=0.36\columnwidth]{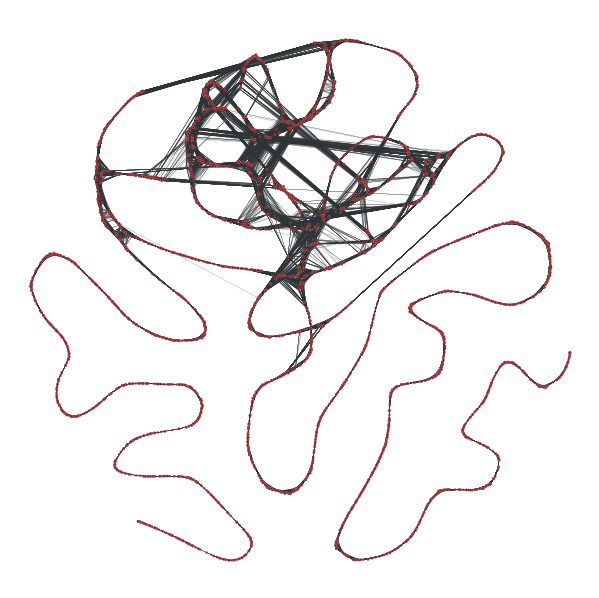} &
    \includegraphics[width=0.36\columnwidth]{figs/plat1919_base.png} &
    \includegraphics[width=0.36\columnwidth]{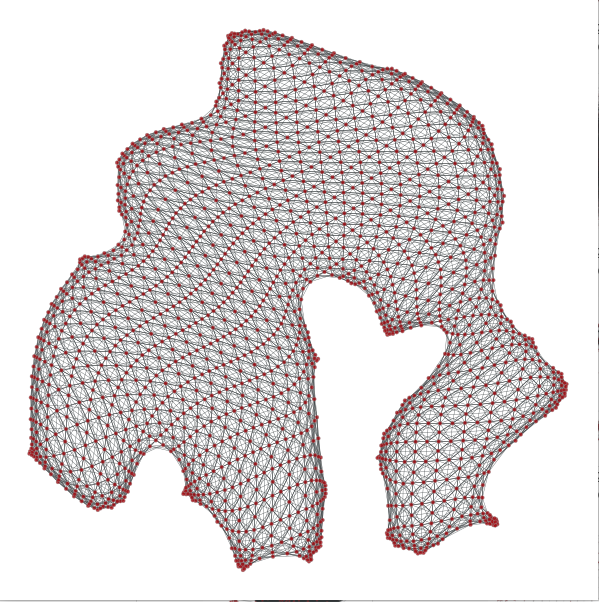} \\ \hline
    \end{tabular}
    \label{tab:tsviscomp}
\end{table*}

\subsection{Comparison Experiments}
\subsubsection{Runtime}

To compare GUMAP with tsNET, we compute the percentage improvement in runtime and metrics as the difference between the values obtained by GUMAP and tsNET divided by the value obtained by tsNET, similar as used with the experiments in Section \ref{sec:eval}. 

\begin{figure}[h!]
    \centering
    \includegraphics[width=\columnwidth]{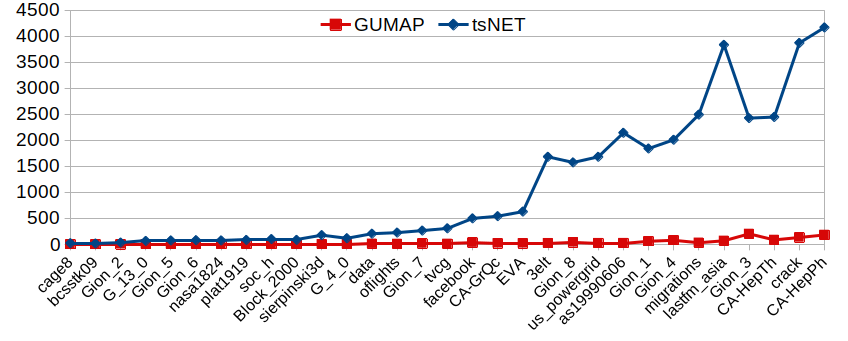}
    \caption{Runtime comparison of GUMAP and tsNET (in seconds): GUMAP runs significantly faster than tsNET.}
    \label{fig:ts_umap_runtime}
\end{figure}

Figure \ref{fig:ts_umap_runtime} shows the runtime of GUMAP and tsNET in seconds, with the graphs sorted by size. Clearly, GUMAP vastly outperforms tsNET on runtime efficiency, running over 90\% faster on average. It can be seen that the growth in runtime is also much less steep than tsNET. This observation is consistent with the faster runtime of UMAP over t-SNE for dimension reduction~\cite{mcinnes2018umap}.

\subsubsection{Quality Metrics}
Figures \ref{fig:ts_umap_np}-\ref{fig:ts_umap_stress} show the quality metric comparison of GUMAP and tsNET. Overall, GUMAP computes similar neighborhood preservation metrics compared to tsNET, as seen in Figure \ref{fig:ts_umap_np}, with on average only a 2\% difference. Similarly, for edge crossing, as seen in Figure \ref{fig:ts_umap_crossing}, GUMAP consistently computes drawings with similar or even slightly lower edge crossing than tsNET, at 8\% better on average.

\begin{figure}[h!]
    \centering
    \includegraphics[width=\columnwidth]{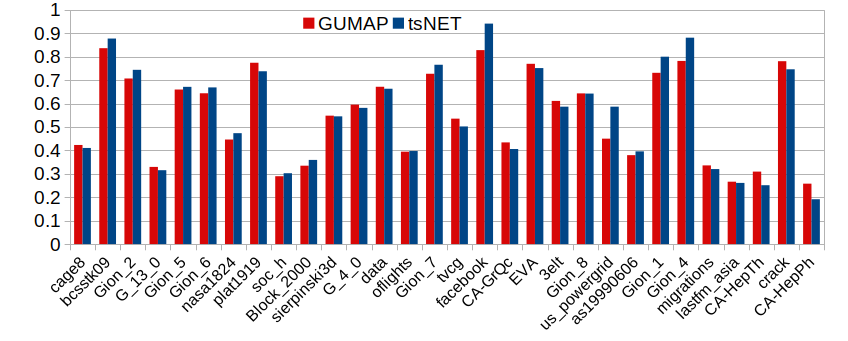}
    \caption{Neighborhood preservation comparison of GUMAP and tsNET: GUMAP obtains similar neighborhood preservation to tsNET.}
    \label{fig:ts_umap_np}
\end{figure}

\begin{figure}[h!]
    \centering
    \includegraphics[width=\columnwidth]{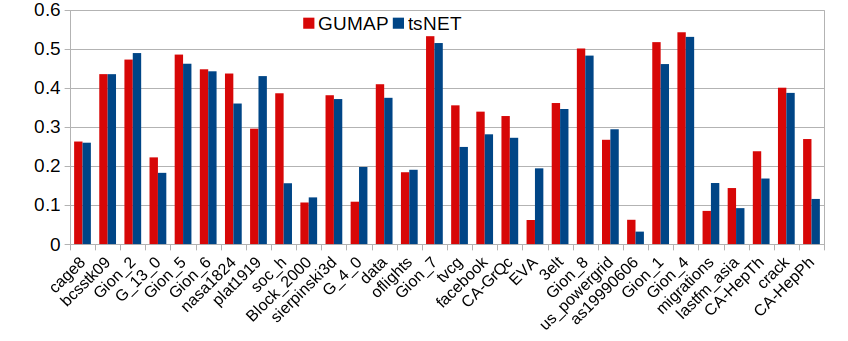}
    \caption{Shape-based metrics comparison of GUMAP and tsNET: on average, cases where GUMAP outperforms tsNET on shape-based metrics outnumber the opposite cases.}
    \label{fig:ts_umap_shp}
\end{figure}

\begin{figure}[h!]
    \centering
    \includegraphics[width=\columnwidth]{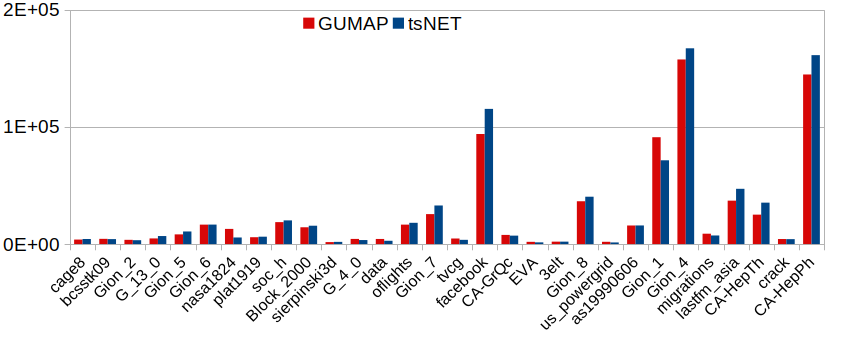}
    \caption{Edge crossing comparison of GUMAP and tsNET: on average, GUMAP obtains 8\% lower edge crossing than tsNET.}
    \label{fig:ts_umap_crossing}
\end{figure}

\begin{figure}[h!]
    \centering
    \includegraphics[width=\columnwidth]{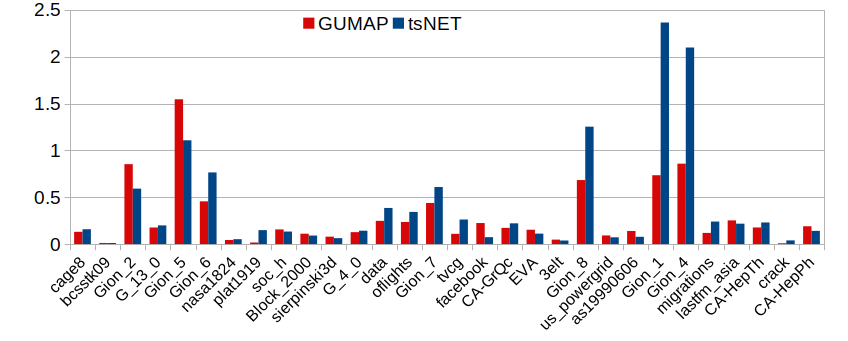}
    \caption{Stress comparison of GUMAP and tsNET: on average, GUMAP obtains 18\% lower stress than tsNET.}
    \label{fig:ts_umap_stress}
\end{figure}

On shape-based metrics, as seen in Figure \ref{fig:ts_umap_shp}, on most graphs, both GUMAP and tsNET compute drawings obtaining similar metrics. On pairwise comparisons of shape-based metrics on drawings of the same graph, GUMAP obtains on average 15\% better metrics over tsNET: cases where GUMAP produces drawings with similar or even higher shape-based metrics, especially on larger scale-free and GION graphs such as Gion\_3, CA-HepTh and CA-HepPh, outnumber cases such as on a few sparser graphs, such as G\_4\_0 and EVA, where tsNET produces drawings with higher shape-based metrics.

On stress, as seen in Figure \ref{fig:ts_umap_stress}, drawings by GUMAP obtain on average 18\% lower stress than tsNET. On most graphs, GUMAP and tsNET produce drawings with similar metrics, with cases where GUMAP significantly outperforms tsNET outnumbering the opposite. The largest improvements are seen on the larger GION graphs with high density and diameters, such as GION\_8, GION\_1, and GION\_4.

\subsubsection{Visual Comparison}

Table \ref{tab:tsviscomp} shows a few examples of visual comparisons between GUMAP and tsNET. In general, both GUMAP and tsNET are able to produce high-quality graph drawings, but with different characteristics. 

On scale-free graphs, such as as19990606 and G\_4\_0, UMAP tends to group vertices closer together; while drawings by tsNET also still show some clustering structures, the clusters tend to be more ``spread out'' compared to the more ``compact'' drawings by GUMAP. Similar patterns can be seen on graphs from other sets, where the drawings by GUMAP are more ``compact'' while tsNET drawings are more ``spread out''. 

While the more compressed drawings by GUMAP might de-emphasize local structures compared to tsNET, this can provide better balance with preserving the global structure compared to the over-spreading of tsNET, such as with the global contour of the mesh being more distorted in plat1919, or with the twisted ``long arms'' in GION\_1.

\subsection{Discussion and Summary}

Our experiments demonstrate the effectiveness of GUMAP for graph drawing by comparison to tsNET, a graph drawing algorithm based on the popular DR algorithm t-SNE, which has been shown to outperform other popular graph drawing algorithms on neighborhood preservation. Figure \ref{fig:ts_umap_avg} shows the average runtime and quality metrics of GUMAP and tsNET, showing the significant 90\% runtime improvement of GUMAP over tsNET with similar neighbourhood preservation and even better on crossing and stress, at 15\% and 18\% better, respectively.

\begin{figure}[h!]
    \centering
    \subfloat[Runtime]{
        \includegraphics[width=0.3\columnwidth]{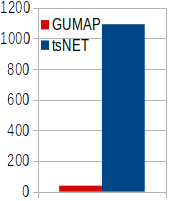}
        \label{fig:runtime_avg}
    }
    \subfloat[Neigh. preserv.]{
        \includegraphics[width=0.3\columnwidth]{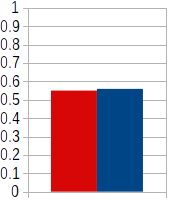}
        \label{fig:rnp_avg}
    }
    \subfloat[Shape-based]{
        \includegraphics[width=0.3\columnwidth]{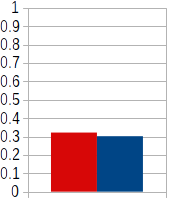}
        \label{fig:shp_avg}
    }
    \qquad
    \subfloat[Crossing]{
        \includegraphics[width=0.3\columnwidth]{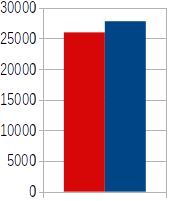}
        \label{fig:crossing_avg}
    }
        \subfloat[Stress]{
        \includegraphics[width=0.3\columnwidth]{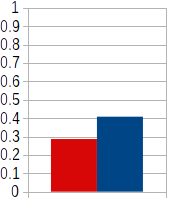}
        \label{fig:stress_avg}
    }
    \caption{Average runtime and quality metrics: GUMAP runs over 90\% faster than tsNET with similar neighborhood preservation and better stress and crossing.}
    \label{fig:ts_umap_avg}
    \vspace{-2mm}
\end{figure}

GUMAP runs significantly faster at over 90\% faster than tsNET, consistent with the significant runtime improvement of UMAP over t-SNE; despite the added $O(nm)$ preprocessing, the $O(n \log n)$ time cost optimization of GUMAP provides speed-ups over tsNET. Drawings by GUMAP obtain the same level of neighborhood preservation as drawings by tsNET, while also obtaining slightly better edge crossings and significantly better shape-based metrics and stress. In addition, visual comparisons show a strength of GUMAP over tsNET, where it avoids the pitfall of tsNET of occasionally over-expanding local vertex neighborhoods to the extent of distorting the overall shape of the graph drawing.

\emph{In summary, GUMAP runs over 90\% faster with the same level of neighborhood preservation as tsNET, as well as better edge crossing and stress, 15\% and 18\% better, respectively.}

\section{Conclusion}
\label{sec:conclusion}

We present a study on fast graph drawing UMAP algorithms. We present three fast UMAP algorithms, which improve the runtime over GUMAP, by utilizing spectral sparsification for sampling a representative subgraph with $O(n \log n)$ edges, partial BFS to reduce the runtime of shortest path and $k$NN graph computations to $O(n)$ time, and edge sampling on the $k$NN graph to reduce the runtime of the gradient descent iterations to sublinear time.

Experiments show that SS-GUMAP, SL-GUMAP, and SSSL-GUMAP achieve significant runtime improvements over UMAP for graph drawing, while maintaining a balance between the runtime gain and preservation of the quality of the drawing.
SS-GUMAP runs 28\% faster than UMAP with almost the same neighborhood preservation, edge crossing, and shape-based metrics, while SL-GUMAP and SSSL-GUMAP run over 80\% faster than UMAP with less than 15\% difference on all quality metrics from UMAP on average.

We also demonstrate the effectiveness of UMAP for graph drawing through the comparison of GUMAP, a direct application of UMAP to graph drawing, to tsNET, where GUMAP runs over 90\% faster than tsNET while obtaining the same level of neighborhood preservation.

Future work includes reducing the runtime of the UMAP graph drawing algorithms to overall sub-linear-time while preserving the quality, to improve scalability to even larger graphs with hundreds of thousands of vertices and edges.
For example, recently sublinear-time force-directed~\cite{meidiana2020sublinear,meidiana2020sublinear-tvcg} algorithms and sublinear-time stress minimization algorithms~\cite{meidiana2021stress} have been successfully developed. 


%




\ifCLASSOPTIONcaptionsoff
  \newpage
\fi



\bibliographystyle{IEEEtran}
\bibliography{IEEEabrv,template}
\end{document}